\documentclass[11pt]{article}
\usepackage[utf8]{inputenc}
\usepackage{amsmath, amssymb, amsthm, dsfont}
\usepackage{enumitem}
\usepackage{graphicx}
\usepackage{subcaption}
\usepackage[dvipsnames]{xcolor}
\usepackage{hyperref, comment}
\usepackage[ruled,vlined,linesnumbered]{algorithm2e}
\usepackage[left=1.00in, right=1.00in, top=1.00in, bottom=1.00in]{geometry}
\usepackage{tikz}
\usetikzlibrary{calc}

\usepackage[framemethod=tikz]{mdframed}
\mdfsetup{
	roundcorner=4pt,
	linewidth=.5,
	innerleftmargin=0pt,
	innerrightmargin=10pt,
	innertopmargin=10pt,
	innerbottommargin=10pt,
	skipabove=12pt,skipbelow=12pt
}

\allowdisplaybreaks

\usepackage{color}

\title{Beyond Worst-Case Budget-Feasible Mechanism Design}
\author{Aviad Rubinstein\thanks{Supported by NSF CCF-1954927, and a David and Lucile Packard Fellowship.}\\Stanford University\\\texttt{aviad@cs.stanford.edu} \and Junyao Zhao\thanks{Supported by NSF CCF-1954927.}\\ Stanford University\\\texttt{junyaoz@stanford.edu}}
\date{}

\def \exp {\textnormal{exp}}
\def \E {\mathbb{E}}
\def \P {\textnormal{Pr}}
\def \R {\mathbb{R}}
\def \eg {{e.g.}}
\def \ie {{i.e.}}

\newcommand{\floor}[1]{\left\lfloor #1 \right\rfloor}

\newtheorem{theorem}{Theorem}[section]
\newtheorem{lemma}[theorem]{Lemma}

\newtheorem{claim}[theorem]{Claim}
\newtheorem{definition}[theorem]{Definition}
\newtheorem{fact}[theorem]{Fact}
\newtheorem{example}[theorem]{Example}
\newtheorem{observation}[theorem]{Observation}

\begin{document}

\maketitle

\begin{abstract}
Motivated by large-market applications such as crowdsourcing, we revisit the problem of budget-feasible mechanism design under a ``small-bidder assumption''. Anari, Goel, and Nikzad (2018) gave a mechanism that has optimal competitive ratio $1-1/e$ on worst-case instances. However, we observe that on many realistic instances, their mechanism is significantly outperformed by a simpler open clock auction by Ensthaler and Giebe (2014), although the open clock auction only achieves competitive ratio $1/2$ in the worst case. Is there a mechanism that gets the best of both worlds, i.e., a mechanism that is worst-case optimal and performs favorably on realistic instances? To answer this question, we initiate the study of {\em beyond worst-case budget-feasible mechanism design}. 

Our first main result is the design and the analysis of a natural mechanism that gives an affirmative answer to our question above:
\begin{itemize}
\item We prove that {\em on every instance}, our mechanism performs at least as good as all uniform mechanisms, including Anari, Goel, and Nikzad's and Ensthaler and Giebe's mechanisms.
\item Moreover, we empirically evaluate our mechanism on various realistic instances and observe that it {\em beats} the worst-case $1-1/e$ competitive ratio {\em by a large margin} and compares favorably to both mechanisms mentioned above.
\end{itemize}

Our second main result is more interesting in theory: We show that in the semi-adversarial model of {\em budget-smoothed analysis}, where the adversary designs a single worst-case market for a distribution of budgets, our mechanism is optimal among all (including non-uniform) mechanisms; furthermore our mechanism guarantees a strictly better-than-$(1-1/e)$ expected competitive ratio for any non-trivial budget distribution {\em regardless of the market}. (In contrast, given any bounded range of budgets, we can construct a single market where Anari, Goel, and Nikzad's mechanism achieves only $1-1/e$ competitive ratio for every budget in this range.) We complement the positive result with a characterization of the worst-case markets for any given budget distribution and prove a fairly robust hardness result that holds against any budget distribution and any mechanism.
\end{abstract}

\setcounter{page}{0}
\thispagestyle{empty}
\newpage

\section{Introduction}
The budget-feasible mechanism design problem was introduced by Singer~\cite{singer2010budget} and has become a core problem in algorithmic mechanism design~\cite{chen2011approximability,DobzinskiPS11,BadanidiyuruKS12,SingerM13,ChanC14,EG14,GoelNS14,HorelIM14,BalkanskiH16,ChanC16,NushiS0K16,ZhaoLM16,ZhengWGZTC17,LeonardiMSZ17,anari2018budget,KhalilabadiT18,AmanatidisKS19,gravin2019optimal,LiZY20,BGGST21}. We will use microtask crowdsourcing as a running example for this problem (see Section~\ref{section:problem_setup} for a formal setup): An employer (buyer) on a crowdsourcing platform (market $I$) such as Mechanical Turk or Microworkers is given a fixed budget $B$, and is looking to acquire some services from a set of workers (sellers) $[n]$. Each worker $i$ can perform a microtask (provide a service) that has a utility $u_i$ to the employer at an incurred cost $c_i$ to the worker himself. The employer's total utility is $\sum_{i\in W}u_i$ for the services provided by each subset of workers $W\subseteq[n]$. As is common in the literature, the employer knows the workers' utilities $u_i$'s (e.g.~by grading their work ex-post, or using the worker's rating on previous tasks), but does not know their private costs $c_i$'s. Moreover, in large-market applications like microtask crowdsourcing, it is often very natural to make a small-bidder assumption: the cost of each worker is a small fraction of the employer's total budget%
\footnote{Other important applications where this assumption is natural include allocation of R\&D subsidies by government agencies and emission reduction auctions~\cite{anari2018budget}.}.

The objective of budget-feasible mechanism design is to design a truthful mechanism that maximizes the employer's total utility while keeping the total payment to the workers within the budget. Roughly speaking, a truthful mechanism makes sure the workers honestly report their private costs $c_i$'s by providing them incentives (payments) and decides which subset of services the employer will get (allocation), and we want the mechanism to maximize the total utility of the services allocated to the employer, under the constraint that the total payment does not exceed $B$.

Without any incentive constraints (i.e., the workers' costs are public, and the employer only needs to pay a worker's cost to get the worker's service), this becomes the well-known knapsack problem. Therefore, it is standard to consider the following performance metric for a mechanism: the ratio between the utility achieved by the mechanism and the optimal utility of the knapsack problem without incentive constraints {\em in a worst-case market and for a worst-case budget}. This metric is called the worst-case competitive ratio, and a mechanism is $\alpha$-competitive if its worst-case competitive ratio is $\ge\alpha\in[0,1]$.

Research in budget-feasible mechanism design has been focusing on designing (polynomial-time) mechanisms that achieve optimal worst-case competitive ratio. Under the small-bidder assumption, ~\cite{anari2018budget} gave a $(1-1/e)$-competitive mechanism and characterized the worst-case instances\footnote{An instance is specified by a market $I$ and a budget $B$.} for which any mechanism can only achieve at most $1-1/e$ competitive ratio.

Although this optimal result provides a satisfactory answer with respect to worst-case competitive ratio, our quest to design even better mechanisms does not come to an end. Indeed, recall that worst-case competitive ratio measures a mechanism's performance in worst-case market given worst-case budget. Such worst-case market and worst-case budget rarely appear in practice. Even if we are given an typical-case market and/or an typical-case budget, a mechanism that achieves optimal worst-case competitive ratio could (potentially) perform as bad as on the worst-case instance. We probably would not prefer such worst-case optimal mechanism over other mechanisms that perform much better on the typical-case instances. To make this point more concrete, consider the following extremely simple instance: 
\begin{example}\label{ex:identical_costs_instance}
\normalfont
The buyer has a budget $B=n$, and each of $n$ seller's services has a cost $1$ to the seller himself and a utility $1$ to the buyer.
\end{example}
For the simple instance in Example~\ref{ex:identical_costs_instance} (which satisfies small-bidder assumption), simply offering a payment of $1$ to each seller extracts full utility, but~\cite{anari2018budget}'s worst-case optimal mechanism only obtains a $(1-1/e)$-fraction. Moreover, instead of identical sellers' costs, consider a more natural variant of Example~\ref{ex:identical_costs_instance}, where the sellers' costs are sampled i.i.d.~from a natural distribution (e.g., Gaussian/uniform/exponential/mixture distribution). Our numerical simulation shows that~\cite{anari2018budget}'s mechanism only obtains close-to-$(1-1/e)$ fraction of the optimal utility for these natural instances, while a simple open clock auction~\cite{EG14}, that is equivalent to setting a single uniform price, obtains significantly larger fractions (see Table~\ref{table:synthetic}). Of course the open clock auction also extracts the full utility for the simple instance in Example~\ref{ex:identical_costs_instance}. However, it is known that the open clock auction is suboptimal in the worst case: \cite{anari2018budget} exhibits a simple example for which the open clock auction has worst-case competitive ratio $1/2$.

\begin{quote}{\em Is there a mechanism that gets the best of both worlds,~i.e., a mechanism that is worst-case optimal and ``performs favorably'' on every instance (not just in the worst case)?}
\end{quote}
By ``perform favorably'', we mean that the mechanism should achieve utility at least as good as a large class of mechanisms. Which class of mechanism should we consider as an appropriate benchmark? At least, we want this class to include the previous mechanisms of~\cite{anari2018budget} and~\cite{EG14}. The most ambitious class is obviously the class of all the mechanisms, but as we now explain, it is unfair to compare with this class. Consider the mechanism in the following example:
\begin{example}\label{ex:silly_mechanism}
\normalfont
Consider an arbitrary instance $(I,B)$ specified by market $I$ and budget $B$, which becomes a knapsack problem when sellers' costs are public, and the optimal solution (i.e., the optimal subset of sellers' services) to this knapsack problem always exists. Now we hard-code the market $I$ in the following mechanism: When given an input instance $(I',B')$ (assume for simplicity\footnote{This is without loss of generality, because otherwise the mechanism could use an arbitrary mapping from sellers in $I$ to sellers in $I'$.} that $I'$ has the same number of sellers as $I$, but the sellers' costs and utilities in $I'$ can be arbitrarily different from $I$), this mechanism reads nothing from input except the budget $B'$, and it always non-uniformly offers each seller, who is in the optimal knapsack solution of instance $(I,B')$, a posted price that is equal to this seller's cost in $I$, and offers nothing to the remaining sellers.
\end{example}
Although the mechanism in Example~\ref{ex:silly_mechanism} is silly (because it always decides the allocation and payments according to $I$ regardless of the actual market $I'$ it is facing), it is a well-defined non-uniform posted price mechanism that is truthful and budget-feasible. Even though we expect this mechanism to perform poorly in general, it is optimal for the specific market $I$ that is hard-coded in it, and there is no way we can compete with such unreasonable mechanism on instance $(I,B)$. In order to exclude such mechanisms while including the mechanisms of~\cite{anari2018budget} and~\cite{EG14}, we restrict our attention to the class of all the uniform mechanisms (for now\footnote{In our results for budget-smoothed analysis, we will compare to all (possibly non-uniform) mechanisms.}). Roughly speaking, a mechanism is uniform if the distributions of normalized offers is essentially the same for all the sellers (see Section~\ref{section:further_important_concepts} for the exact definition).

It is noteworthy that unlike algorithm design, where one can combine two algorithms by taking the best solution outputted by these algorithms, naively combining two mechanisms in such way typically does not result in a truthful mechanism, which motivates us to search for a new mechanism that satisfies the desiderata in our main question.

With the above motivation, we initiate the study of {\em beyond worst-case budget-feasible mechanism design}, and we also make the small-bidder assumption given its wide applicability in practice (see~\cite[Section 10]{anari2018budget}). In the next two subsections, we give an overview of our results. In terms of the significance, we believe the first main result (instance optimality), which compares our new mechanism with uniform mechanisms, is more significant from the practical perspective, and the techniques are arguably not complicated and hence can be applied in practice. The second main result (budget-smoothed analysis), which compares our new mechanism with the general (possibly non-uniform) mechanisms is more interesting from the theoretical perspective. We believe these two results complement each other, and we hope these results could encourage researchers in the broad area of mechanism design to examine the worst-case optimal mechanisms for their mechanism design problems through beyond-worst-case lens and design even better mechanisms with improved beyond-worst-case performance.

\subsection{Main result I: instance optimality}
The first result of this paper is the design and (theoretical and empirical) analysis of a new natural mechanism. We prove that our new mechanism performs at least as good as any uniform mechanism on every instance (Theorem~\ref{thm:random_sample_greedy}).
\begin{theorem}[Instance-optimality against uniform mechanisms] \hfill\\
We give a computationally efficient,
truthful and strictly budget-feasible randomized mechanism, that, on  every instance of the budget-feasible mechanism design problem with additive buyer's utility function and small sellers, achieves ${\ge (1-o(1))}$ of the expected utility of any uniform mechanism.
\end{theorem}

Moreover, we empirically evaluate our mechanism on many realistic instances and observe that it beats the worst-case $1-1/e$ competitive ratio by a large margin.
\paragraph{Empirical analysis}
Specifically, we compare the performance of our mechanism, the open clock auction~\cite{EG14}, and~\cite{anari2018budget}'s mechanism on synthetic instances (see Section~\ref{section:sythetic} for details).
We observe that our mechanism and the open clock auction outperform~\cite{anari2018budget}'s worst-case optimal mechanism on all synthetic instances by a large margin.
In the instances where the distribution of sellers' cost-per-utility is multi-modal\footnote{We note that multi-modal sellers' distribution is possible in the real world. For example, in an international market, the average cost-per-utility of sellers in a country could differ from that in another country because of the difference of resources/technology between different countries.}, our mechanism outperforms both other mechanisms significantly (recall that we indeed prove that it is always optimal).

\subsubsection*{Our mechanism in a nutshell}
An idealized version of our mechanism, where we know the market statistics (i.e., the empirical distribution of sellers' types\footnote{A seller's type is specified by his cost and utility.}), has the following nice interpretation: each seller is independently offered one of two possible prices, and can choose to accept or reject the offer she receives.
Knowing the market statistics is a reasonable assumption in many cases in practice, e.g.~when the buyer has access to historical bids. 
In general, when the statistics are not known, we can randomly partition the sellers into two subsets, and compute prices for each half based on market statistics estimated from truthful reporting of costs from the other half.

The main novelty of our mechanism is the design of its idealized version---a greedy-type uniform ``mechanism'' (Mechanism~\ref{mech:greedy}), which can be interpreted as a probabilistic combination of at most two uniform prices per utility. We prove that this greedy ``mechanism'' is instance-optimal compared to all the uniform mechanisms by a neat greedy exchange argument. We are surprised that despite being such a natural ``mechanism'' (from the information-theoretic point of view), Mechanism~\ref{mech:greedy} has never been studied in the literature to our best knowledge.

In the random partitioning step for estimating the market statistics, the technical part is how to control the noise caused by random partitioning (overly large noise could ruin the budget feasibility of the mechanism without giving up a significant fraction of utility). Thanks to the simple form of our idealized mechanism, we are able to succinctly discretize the space of candidate allocation rules. Moreover, in order to upper bound the influence of each individual seller during the random partitioning, we truncate the allocation rule, which does not lose much utility because of small-bidder assumption. By a careful probabilistic analysis, we show that combining these techniques is sufficient to approximate our idealized greedy mechanism within negligible error. Therefore, the approximate version of the greedy mechanism, which is our final mechanism (Mechanism~\ref{mech:RS_greedy}), is (nearly) optimal on every instance compared to any uniform mechanism.

\subsection{Main result II: budget-smoothed analysis}
We have shown that empirically our mechanism's performance on realistic instances is much better than the $1-1/e$ competitive ratio on the ``worst-case instance'', which suggests that optimality on the ``worst-case instance'' is a weak notion that fails to capture better-than-worst-case performance. We also have shown that our mechanism beats all the uniform mechanisms on ``every instance'', but as we explained before, we restricted our attention to the class of uniform mechanisms, because it is unreasonable to compare with the class of non-uniform mechanisms on ``every instance'', which suggests that optimality on ``every instance'' is somewhat too strong if we hope to compare our mechanism with the more general non-uniform mechanisms.

Thus in addition to our first result, we strike a reasonable middle ground between ``worst-case instance'' and ``every instance'' by examining our mechanism under the {\em budget-smoothed analysis} framework recently introduced in~\cite{RZ20} in the context of submodular maximization. This framework gives a reasonable notion of beyond-worst-case instances that allows us to theoretically compare our mechanism to all the (even non-uniform) mechanisms.

Briefly (see the formal definition in Section~\ref{sec:budget-smoothed-model}), the budget-smoothed analysis framework is a semi-adversarial model: We first pick a budget distribution and a mechanism, and then the adversary, who knows the mechanism and the budget distribution, chooses a single worst-case market, and finally we sample a budget from the distribution and measure the mechanism's {\em expected competitive ratio} (formally defined in Section~\ref{sec:budget-smoothed-model}) on the adversarially chosen market, where the expectation is over the randomness of the budget distribution and (potentially) the mechanism itself. (The motivation of the budget-smoothed analysis in the context of budget-feasible mechanism design deserves an in-depth discussion, which we defer to Section~\ref{sec:smoothed_motivation} due to the interest of space.)

We show the following fundamental results in the budget-smoothed analysis model.

\subsubsection*{Optimal mechanism and worst-case markets for any budget distribution}
We prove that our mechanism is optimal (see Definition~\ref{def:budget_smoothed_optimal}) among all the (not necessarily uniform) mechanisms on the worst-case market for any budget distribution\footnote{It is particularly interesting that our mechanism, which does not require any knowledge of the budget distribution, is optimal even when compared with the mechanisms that know the budget distribution. In other words, our mechanism intrinsically adjusts itself to the budget distribution optimally.} (Theorem~\ref{thm:worst_distribution}), and moreover, the expected competitive ratio of our mechanism is guaranteed to be strictly better than $1-1/e$ for every nontrivial budget distribution regardless of the market\footnote{Moreover, in Section~\ref{section:numerical_computing_ratios}, we formulate a (non-convex) mathematical program that computes the expected competitive ratio on the worst-case market for any given budget distribution. We solve this program for various distributions and observe nonnegligible improvement over $1-1/e$.} (Theorem~\ref{thm:beating}). In contrast, given any bounded range of budgets, we construct a {\em single} market where~\cite{anari2018budget}'s worst-case optimal mechanism cannot beat the worst-case $1-1/e$ competitive ratio {\em for any budget in this range} (Theorem~\ref{thm:hard_instance_for_agn}), which exhibits a strong separation between our mechanism and ~\cite{anari2018budget}'s mechanism.

Our proof of the optimality result is conceptually appealing: We observe that once we fix an arbitrary budget distribution, determining the worst-case market and the optimal mechanism is a min-max game between the adversary and the mechanism designer, in which the adversary tries to give a market that minimizes the mechanism's performance, and the mechanism designer hopes to design a mechanism that performs best on the adversarially chosen market. We analytically solve the equilibrium for this min-max program, and the solution comes with a characterization of the worst-case markets for any given budget distribution (Theorem~\ref{thm:worst_distribution}).

\subsubsection*{Robust hardness result against any budget distribution}
On the negative side, we prove a robust hardness result that shows for any budget distribution and any mechanism, there is a market on which the mechanism's expected competitive ratio is bounded away from 1 (specifically, at most 0.854 --- see Theorem~\ref{thm:lower_bound}). In comparison, the previous worst-case hardness result~\cite{anari2018budget} is very sensitive to budget perturbation: If we perturb (i.e., multiply) the budget of the worst-case instance by a significant factor like $2.5$, then simply setting a single uniform price (i.e.,~\cite{EG14}'s open clock auction) will achieve $100\%$ of the optimal utility.
\section{Preliminaries}\label{section:preliminaries}
\subsection{Problem setup}\label{section:problem_setup}
In the budget-feasible mechanism design (a.k.a., procurement auction) problem with additive utility, there is a market $I$ consisting of one buyer and $n$ sellers, and each seller $i$ has an item with a public utility $u_i\in\mathbb{R}_{\ge0}$ and a private cost $c_i\in\mathbb{R}_{\ge0}$. The buyer has a budget $B\in\mathbb{R}_{\ge0}$ and wants to buy items from the sellers. The goal of the budget-feasible mechanism design problem is to design a truthful mechanism that maximizes buyer's total utility while keeping the total payment to sellers within the budget, which we now explain more formally.

\paragraph{Truthful mechanisms} A mechanism takes as input the buyer's budget $B$, the sellers' public utilities $u_i$'s and the private costs\footnote{We note that there are mechanisms that do not directly ask the sellers to report their costs such as clock auctions. However, our definition is without loss of generality by the revelation principle.} $c_i$'s reported by the sellers, and then outputs which items should be allocated to the buyer and how much the buyer should pay to each seller. Formally, the output of a (randomized) mechanism, i.e., allocation and payments, can be represented by\footnote{If the mechanism is deterministic, $g$ and $Q_{g}$ output the deterministic allocations and deterministic payments, respectively, and if the mechanism is randomized, they output the expected allocations and expected payments.} an \textit{allocation function} $g:\R_{\ge0}^n\to[0,1]^n$ and a \textit{payment function} $Q_{g}:\R_{\ge0}^n\to\R_{\ge0}^n$, where $g$ takes the sellers' cost-per-utility $\gamma_i:=c_i/u_i$'s as input $\vec{\gamma}$, and outputs (the expectation of) the fraction of each item that is allocated to the buyer (and hence, the expected utility the buyer gets from seller $i$ is the $i$-th coordinate of the output of $g$, which we denote by $g(\vec{\gamma})_i$, times $u_i$, and the expected cost of seller $i$ is $g(\vec{\gamma})_i\cdot c_i$), and $Q_{g}$ takes the same input and outputs the associated (expected) payment-per-utility for each item (namely, the expected payment to seller $i$ is the $i$-th coordinate of the output of $Q_{g}$, which we denote by $Q_{g}(\vec{\gamma})_i$, times $u_i$).

A deterministic mechanism is \textit{truthful} if reporting the true $\gamma_i$ always maximizes the net profit for each seller $i\in[n]$, namely, for any $\gamma_{-i}\in\mathbb{R}_{\ge0}^{n-1}$ (where $\gamma_{-i}$ denotes all the $\gamma_j$'s except $\gamma_i$), for all $z\in\mathbb{R}_{\ge0}$,
\begin{equation}\label{eq:truthfulness}
    Q_{g}(\gamma_i,\gamma_{-i})_i\cdot u_i-g(\gamma_i,\gamma_{-i})_i\cdot c_i\ge Q_{g}(z,\gamma_{-i})_i\cdot u_i-g(z,\gamma_{-i})_i\cdot c_i.
\end{equation}

In general, a mechanism can be randomized, and a \textit{randomized mechanism} is simply a distribution of deterministic mechanisms. In this light, we say a randomized mechanism is \textit{truthful-in-expectation} if reporting the true $\gamma_i$ only maximizes seller $i$'s net profit in expectation over the randomness of the mechanism, i.e., Eq.~\eqref{eq:truthfulness} holds in expectation for the randomized mechanism.

The celebrated Myerson's lemma~\cite{myerson1981optimal} asserts that (i) an allocation function $g$ can be implemented as a truthful-in-expectation mechanism if and only if $g$ is \textit{monotone}, i.e., for all $i\in[n]$ and any $\gamma_{-i}\in\mathbb{R}_{\ge0}^{n-1}$, $g_i(\cdot):=g(\cdot,\gamma_{-i})_i$ is a non-increasing function, and (ii) there exists a unique payment function $Q_{g_i}$ associated with $g_i$, which is given by $Q_{g_i}(\gamma):=\gamma\cdot g_i(\gamma)+\int_{\gamma}^\infty g_i(z)dz$.

\paragraph{Budget feasibility} Note that we want the randomized mechanisms to \textit{strictly satisfy} the budget constraint, i.e., every deterministic mechanism in the support of the distribution has to satisfy the following budget constraint
$$\sum_{i\in[n]}Q_{g_i}(\gamma_i)u_i\le B,$$
and our proposed randomized mechanism will indeed strictly satisfy the budget constraint.

The goal of budget-feasible mechanism design is to design a (randomized) mechanism, that is (truthful/truthful-in-expectation) and budget-feasible, to maximize the buyer's (expected) total utility $\sum_{i\in[n]}g_i(\gamma_i)u_i$.

If the sellers' costs are public, the problem becomes the well-known knapsack problem, and we call the optimal utility of this knapsack problem the {\em non-IC} (\ie, without the incentive compatible constraints) optimal utility. The standard performance measure for a (randomized) mechanism $\mathcal{M}$ on the instance $(I,B)$ is the {\em competitive ratio}, i.e.~the ratio $\mathcal{R}_{\mathcal{M}}(I, B)$ between the (expected) total utility (over $\mathcal{M}$'s randomness) achieved by $\mathcal{M}$ and the non-IC optimal utility.

Finally, we make a \textit{small-bidder assumption}~\cite{anari2018budget}: for budget $B$, we require that each seller's cost is at most $o(B)$.

\subsubsection{Further important concepts}\label{section:further_important_concepts}
\paragraph{Uniform mechanism} We call a mechanism\footnote{For example, the idealized versions of~\cite{anari2018budget}'s mechanism (Mechanism~\ref{mech:AGN}), \cite{EG14}'s mechanism and our mechanism (Mechanism~\ref{mech:greedy}) are all uniform mechanisms.} with allocation function $g$ \textit{uniform} if given any $\gamma_i$'s, there exists a 1-dimensional allocation function $f:\R_{\ge0}\to[0,1]$ such that for all $i\in[n]$, it holds that $g_i(\cdot)=f(\cdot)$. Otherwise, we call the mechanism \textit{non-uniform}.

\paragraph{Fractional versus indivisible}
We mentioned that the allocation function specifies the fraction of item purchased from each seller. This makes sense when the item is fractional, \eg, the item is the time of a worker. However, there are settings where the items are indivisible, and then, the image of an allocation function should be $\{0,1\}^n$ instead. Under small-bidder assumption, an indivisible item procurement problem can be reduced to a fractional problem. Specifically, there is a rounding procedure from~\cite[Supplemental Material, Section 7]{anari2018budget} that we can directly apply.
\begin{lemma}[{\cite[Supplemental Material, Section 7]{anari2018budget}}]\label{lem:agn_rounding}
Let $\tilde{x}_1,\dots,\tilde{x}_n$ be the fractional allocations and $\tilde{p}_1,\dots,\tilde{p}_n$ be the associated payments. Under small-bidder assumption, there is a rounding procedure that outputs integral allocations $x_1,\dots,x_n$ and payments $p_1,\dots,p_n$, which achieves approximately the same expected utility as the fractional allocations, while preserving individual rationality, truthfulness in expectation, and strict budget feasibility.
\end{lemma}

Henceforth, given this reduction, \textbf{unless specified otherwise, we only consider divisible items} in this paper, and the results apply to indivisible items as well.

\subsection{Budget-smoothed analysis}\label{sec:budget-smoothed-model}
Budget-smoothed analysis is a semi-adversarial model introduced in~\cite{RZ20} in the context of submodular optimization. In our setting, given any fixed distribution of budgets $\mathcal{D}$, the performance metric for a mechanism $\mathcal{M}$ in the budget-smoothed analysis is the \textit{$\mathcal{D}$-budget-smoothed competitive ratio}: the worst possible ratio between the utility achieved by $\mathcal{M}$ and the non-IC optimum in expectation (over budget distribution and mechanism’s randomness), \ie,
\[
    \min_{I} \underset{B\sim \mathcal{D}}{\mathbb{E}}[\mathcal{R}_{\mathcal{M}}(I, B)],
\]
where $\underset{B\sim \mathcal{D}}{\mathbb{E}}[\mathcal{R}_{\mathcal{M}}(I, B)]$ is the \textit{expected competitive ratio} of $\mathcal{M}$ on market $I$ for budget distribution $\mathcal{D}$.
Fixing an arbitrary budget distribution $\mathcal{D}$, the goal of the mechanism designer is to design a mechanism $\mathcal{M}$ that achieves optimal $\mathcal{D}$-budget-smoothed competitive ratio, and hence, we have a max-min game between the mechanism designer and the adversary
\[
    \max_{\mathcal{M}}\min_{I} \underset{B\sim \mathcal{D}}{\mathbb{E}}[\mathcal{R}_{\mathcal{M}}(I, B)].
\]
In other words, we are interested in the expected outcome of the following budget-smoothed analysis game:
\begin{mdframed}
\begin{center}
{ Budget-smoothed analysis game}\\[3mm]
\end{center}
\begin{enumerate}
    \item Fix a distribution of budgets $\mathcal{D}$. The mechanism designer, who knows the budget distribution $\mathcal{D}$, picks a mechanism\footnote{Note that the mechanism designer knows $\mathcal{D}$ and hence is allowed to choose a mechanism $\mathcal{M}$ that is tailored to $\mathcal{D}$, i.e., the mechanism designer knows $\mathcal{D}$ and then specifies what $\mathcal{M}$ does for each budget $B$ in the support of $\mathcal{D}$ as she likes. Interestingly, as we will show later, our optimal mechanism does not need any knowledge of $\mathcal{D}$.} $\mathcal{M}$.
    \item The adversary, who knows the budget distribution $\mathcal{D}$ and the mechanism $\mathcal{M}$ chosen by the mechanism designer, chooses a worst-case market\footnote{Note that the adversary chooses the market $I$ after knowing $\mathcal{D}$ and $\mathcal{M}$. For example, if the mechanism designer chose the silly mechanism in Example~\ref{ex:silly_mechanism} that hard-codes some market $I_1$, this mechanism will perform poorly in the budget-smoothed analysis game, because the adversary can choose a completely different market $I_2$ after observing the mechanism chosen by the mechanism designer.} $I$ (sellers' costs and utilities).
    \item Then, a budget $B$ is drawn at random from $\mathcal{D}$.
    \item Finally, the mechanism designer runs $\mathcal{M}$ on the instance $(I,B)$ (and compare the performance to the non-IC optimum).
\end{enumerate}
\end{mdframed}
\begin{definition}\label{def:budget_smoothed_optimal}
A mechanism $\mathcal{M}^*$ is \textit{worst-case optimal for a budget distribution $\mathcal{D}$} if for any other mechanism $\mathcal{M}$, $\min_{I} \underset{B\sim \mathcal{D}}{\mathbb{E}}[\mathcal{R}_{\mathcal{M}^*}(I, B)]\ge \min_{I} \underset{B\sim \mathcal{D}}{\mathbb{E}}[\mathcal{R}_{\mathcal{M}}(I, B)].$
\end{definition}

We refer the interested readers to~\cite{RZ20} for the original motivation and naming of the budget-smoothed analysis model.
In our context, we can think of the $\mathcal{D}$-budget-smoothed competitive ratio as the average competitive ratio of multiple employers operating in the worst-case market $I$ with different budgets (the empirical distribution of their budgets is $\mathcal{D}$), and the employers' budgets can easily vary by an order of magnitude because of different sizes of business (as in the Microworkers example in Table~\ref{table:data}). Given such budget distribution supported on a wide range, even if the market $I$ is worst-case, the ``average'' employers who use an ``average'' budget could potentially enjoy a competitive ratio that is significantly better than the worst-case optimal competitive ratio $1-1/e$ (and by Markov inequality, most employers achieve strictly better-than-$(1-1/e)$ competitive ratio).

\section{Instance-optimality against uniform mechanisms}~\label{section:instance_optimal}
In this section, we first derive a uniform ``mechanism'' in the complete-information setting, where the sellers' private costs are known. To be precise, the complete-information uniform ``mechanism'' applies a single monotone allocation function and the associated Myerson's payment function to all the sellers and guarantees strict budget-feasibility just like a normal uniform budget-feasible mechanism, and the only caveat is that to compute the allocation function, the complete-information uniform ``mechanism'' needs to know all the sellers' costs. This complete-information uniform ``mechanism'' is essentially a greedy procedure. Then, we show that this greedy ``mechanism'' is \textit{instance-optimal} compared to all the uniform budget-feasible mechanisms.
That is, for every market and every budget, compared to all the uniform budget-feasible mechanisms that satisfy Myerson's characterization of truthful-in-expectation mechanisms, the greedy ``mechanism'' achieves the optimal buyer's utility.

Apparently, this greedy ``mechanism'' by itself is not very useful, since we eventually want a normal mechanism that works in the setting where sellers' private costs are hidden. Therefore, we design a randomized mechanism to approximate the greedy ``mechanism'', i.e., our randomized mechanism is nearly as good as the greedy ``mechanism'' on every instance. On a high level, this is done by first randomly partitioning the market into two halves, and then applying our greedy ``mechanism'' on each half to get an allocation function and the associated payment function, and finally applying the allocation and payment function we get from one half to the other half in a sequential fashion until certain budget threshold is met.

\subsection{Greedy is an instance-optimal uniform ``mechanism''}
In this subsection, we describe the greedy ``mechanism'' $\textsc{Greedy}$ in the complete information setting, where the sellers' private costs are given, and prove that it is instance-optimal compared to all the uniform truthful-in-expectation budget-feasible mechanisms. The pseudo-code of \textsc{Greedy} is given in Mechanism~\ref{mech:greedy}.

It works as follows -- Suppose that the sellers are grouped and sorted according to their cost-to-utility ratio $c/u$. \textsc{Greedy} searches for the best monotone allocation function (and given the allocation function, the payment is determined by Myerson's lemma). It does this iteratively. In each iteration, suppose that it has bought all the items from the sellers with cost-to-utility ratio at most $c_{i-1}/u_{i-1}$, then it will choose the sellers whose $c/u$ ranges from $(c_{i-1}/u_{i-1})_+$ to  $c_j/u_j$ for some $j\ge i$, and simultaneously increase the fraction bought from these sellers until either they are fully purchased or the budget is exhausted. 

We now explain how \textsc{Greedy} selects the next $j$ in each iteration. 
For each candidate seller $k$, it computes the \textit{marginal utility per marginal payment} (denoted by $e_{i,k}$) achieved by simultaneously increasing the fraction of items purchased from all the buyers in $i,\dots,k$.
\textsc{Greedy} then greedily selects $j$ to be the index that maximizes the marginal utility per marginal payment. 

\begin{algorithm}[ht]
\SetAlgoLined
\SetKwInOut{Input}{Input}
\SetKwInOut{Output}{Output}
\Input{$(c_i,u_i)$ for $i\in[n]$, $B$.}
\SetAlgorithmName{Mechanism}~~
Merge the sellers with equal cost-to-utility $\gamma_i:=\frac{c_i}{u_i}$ into one seller by summing up their costs and utilities, and let $u_0$ be the utility of the merged seller with cost $0$ ($u_0=0$ if there is no such seller) and let $\gamma_0=0, f(0)=1$. Sort all non-zero-cost merged sellers such that the $\gamma_i$ are non-decreasing, and let $n'$ be the number of non-zero-cost merged sellers\;

$i \leftarrow 1$\;
 \While{$i\le n'$}{
  Choose $j\in\{i,i+1,\dots,n'\}$ that maximizes $e_{i,j}$ ($e_{i,j}$ is defined by Eq.~\eqref{eq:e_i_j})\;
  \eIf{$B>q_{i,j}^{\max}$ \textnormal{($q_{i,j}^{\max}$ is defined in Eq.~\eqref{eq:q_i_j_max})}}{
   Let $f(\gamma)=1$ for all $\gamma\in(\gamma_{i-1},\gamma_j]$ and $B=B-q_{i,j}^{\max}$\;
   $i=j+1$\;
   }{
   Let $f(\gamma)=\frac{B}{q_{i,j}^{\max}}$ for all $\gamma\in(\gamma_{i-1},\gamma_j]$ and break\;
  }
 }
 \For{each original seller $i\in[n]$}{
  Purchase $f(\gamma_i)$ fraction of seller $i$'s item and pay $u_i Q_f(\gamma_i)$, where $Q_f$ is the payment rule corresponding to $f$ given by Myerson's lemma, i.e., $Q_f(\gamma)=\gamma\cdot f(\gamma)+\int_{\gamma}^\infty f(z)dz$\;
 }
 \caption{\textsc{Greedy}}
 \label{mech:greedy}
\end{algorithm}
\begin{theorem}\label{thm:greedy_instance_optimal}
For divisible items, \textsc{Greedy} decides the allocation and the payment for all the sellers using a single monotone allocation function and the associated Myerson's payment function, and it is strictly budget-feasible, and moreover, on every instance, \textsc{Greedy} achieves buyer's utility no less than any uniform truthful-in-expectation budget-feasible mechanism.
\end{theorem}
\begin{proof}
First, observe that $f$ in \textsc{Greedy} is a non-increasing function, and we apply the same $f,Q_f$ to all the sellers in \textsc{Greedy}.
In each iteration of the while loop, suppose we increase the allocation function $f$'s value over $(\gamma_{i-1},\gamma_j]$ from zero to certain $f(\gamma_j)$, the payment-per-utility $Q_f(\gamma)=\gamma\cdot f(\gamma)+\int_{\gamma}^\infty f(z)dz$ should also increase for every $\gamma\le\gamma_j$. Specifically, for every $\gamma\le\gamma_{i-1}$, the $\gamma\cdot f(\gamma)$ part does not change, but the $\int_{\gamma}^\infty f(z)dz$ part increases from zero to $\int_{\gamma_{i-1}^+}^{\gamma_j} f(\gamma_j)dz=f(\gamma_j)\cdot(\gamma_j-\gamma_{i-1})$, and thus, $Q_f(\gamma)$ increases by $f(\gamma_j)\cdot(\gamma_j-\gamma_{i-1})$. For every $\gamma\in(\gamma_{i-1},\gamma_j]$, the $\gamma\cdot f(\gamma)$ part increases from zero to $\gamma\cdot f(\gamma_j)$, and the $\int_{\gamma}^\infty f(z)dz$ part increases from zero to $\int_{\gamma}^{\gamma_j} f(\gamma_j)dz=f(\gamma_j)\cdot(\gamma_j-\gamma)$, and thus, $Q_f(\gamma)$ increases by $f(\gamma_j)\cdot\gamma_j$ in total. Since the total utility of the sellers with cost-per-utility at most $\gamma_{i-1}$ is $\sum_{0\le l\le i-1}u_l$, and the total utility of the sellers with cost-per-utility in $(\gamma_{i-1},\gamma_j]$ is $\sum_{i\le l\le j}u_l$, it follows that the additional payment the mechanism makes in this iteration is
$q_{i,j}:=f(\gamma_j)\cdot(\gamma_j-\gamma_{i-1})\cdot\sum_{0\le l\le i-1}u_l+f(\gamma_j)\cdot\gamma_j\cdot\sum_{i\le l\le j}u_l$, which is at most (equal when $f(\gamma_j)=1$)
\begin{equation}\label{eq:q_i_j_max}
q_{i,j}^{\max}:=(\gamma_j-\gamma_{i-1})\cdot\sum_{0\le l\le i-1}u_l+\gamma_j\cdot\sum_{i\le l\le j}u_l,
\end{equation}
and hence, the if condition in \textsc{Greedy} ensures the budget feasibility. Moreover, observe that the additional utility \textsc{Greedy} gains in this iteration is $v_{i,j}:=f(\gamma_j)\cdot\sum_{i\le l\le j}u_l$.
Therefore, the ratio between the additional utility we gain and the additional price we pay, when we increase $f(\gamma)$ uniformly for all $\gamma\in(\gamma_{i-1},\gamma_j]$, is
\begin{equation}\label{eq:e_i_j}
    e_{i,j}:=\frac{v_{i,j}}{q_{i,j}}.
\end{equation}
In each iteration, \textsc{Greedy} selects the best $j$ that maximizes $e_{i,j}$. Now we show the instance optimality using a greedy exchange argument. Consider any other monotone allocation rule $g$ and suppose $\gamma_{i+1}$ is the smallest among all the sellers' $\gamma$'s such that $g(\gamma_{i+1})\neq f(\gamma_{i+1})$. (Such $\gamma_{i+1}$ cannot be $0$ because otherwise, letting $g(0)=1$ cannot increase the payment or decrease the utility for $g$.) Now we show how to make $g$ more consistent with $f$ without decreasing its achieved utility.

\paragraph{Case (i): $g(\gamma_{i+1})>f(\gamma_{i+1})$} Then $f(\gamma_{i+1})<1$ since $g(\gamma_{i+1})\le 1$. We now argue that $f(\gamma_i)=1$. By our choice of $\gamma_{i+1}$, $f(\gamma_i)=g(\gamma_i)\ge g(\gamma_{i+1})>f(\gamma_{i+1})$, and given that $f(\gamma_i)>f(\gamma_{i+1})$, \textsc{Greedy} prefers the items before $i+1$. Therefore it will not start buying the $(i+1)$-th item until those items are exhausted. Moreover, $f(\gamma_{i+1})$ must be strictly positive, because otherwise, $f$ does not spend as much budget as $g$.  Hence indeed $f(\gamma_i)=1$.

Hence \textsc{Greedy} must have chosen the best $e_{i+1,k}$ for some $k>i+1$, where the inequality is due to the budget feasibility of $g$. (Indeed, if $k=i+1$, then there is enough budget for \textsc{Greedy} to increase $f(\gamma_{i+1})$ to $g(\gamma_{i+1})$, since $g$ is budget-feasible.)
Let $\gamma\ge \gamma_{i+1}$ denote the largest cost-per-utility such that $g(\gamma)>0$. We can assume $\gamma=\gamma_l$ for some $l\ge i+1$ because otherwise we can truncate the extra part of $g$ while preserving its utility. Note that \textsc{Greedy} guarantees that $e_{i+1,k}\ge e_{i+1,l'}$ for all $i+1\le l'\le l$. Hence, if we decrease $g$ over $(\gamma_i,\gamma_l]$ to $0$ and use the saved budget to uniformly increase $g$ over $(\gamma_i,\gamma_k]$, the resulting utility cannot decrease. 

\paragraph{Case (ii): $g(\gamma_{i+1})<f(\gamma_{i+1})$} Suppose that \textsc{Greedy} chose the best $e_{i_1,i_2}$ for some $i_1\le i+1\le i_2$. Therefore, $f$ is a constant on $(\gamma_{i_1-1},\gamma_{i_2}]$, and by monotonicity of $g$ and our assumption that $\gamma_{i+1}$ is the first place where two allocation functions differ, it follows that $f$ is strictly larger than $g$ on $(\gamma_{i_1-1},\gamma_{i_2}]$. Since \textsc{Greedy} guarantees that $e_{i_1,i_2}\ge e_{i_1,j}$ for any $j\ge i_1$, we can decrease $g$ on $(\gamma_{i_1-1},+\infty)$ simultaneously and use the saved budget to uniformly increase $g$ on $(\gamma_{i_1-1},\gamma_{i_2}]$, which can not decrease the achieved utility. We keep doing this unless $g$ reaches $1$ on $(\gamma_{i_1-1},\gamma_{i'}]$ for some $i'\le i_2$. Then, either $f$ is $1$ on $(\gamma_{i_1-1},\gamma_{i_2}]$, and hence, $g$ becomes more consistent with $f$, or $f$ is $<1$ on this interval, in which case, we can decrease $g$ on $(\gamma_{i_1-1},+\infty)$ to $0$ and use the saved budget to uniformly increase $g$ on $(\gamma_{i_1-1},\gamma_{i_2}]$.
\end{proof}
Since~\cite{anari2018budget} showed for a single budget, the uniform mechanism given by Mechanism \ref{mech:AGN} (also with knowledge of all $c_i$'s) has worst-case competitive ratio $1-1/e$ (and there is a matching hardness result), Theorem~\ref{thm:greedy_instance_optimal} implies that \textsc{Greedy} has worst-case competitive ratio $1-1/e$.

\subsection{Greedy allocation rule: a lottery of two posted prices}
Before we present the final randomized mechanism, we observe some nice properties of \textsc{Greedy} which will help us analyze \textsc{Greedy} in a more intuitive way. The key observation, which follows directly from the design of \textsc{Greedy}, is that the allocation rule of \textsc{Greedy} has a simple form that can be fully characterized by three parameters\footnote{One might notice that this characterization actually captures a strictly more general class of allocation rules than just the possible outputs of \textsc{Greedy}. This is for the convenience of analysis later, and we will call any allocation rule that can be characterized in this way a ``greedy allocation rule''.}:
\begin{observation}\label{obs:two_step}
The allocation rule $f$ of \textsc{Greedy} is a $(\le 2)$-step function, i.e.,~there exists some $(t, p_1, p_2)$ where $t\in[0,1)$ and\footnote{$0^-$ denotes a strictly negative number that is arbitrarily close to $0$.} $0^-\le p_1\le p_2$,
$$
f\left(\frac{c}{u}\right)=\begin{cases}
1 &   \frac{c}{u} \leq p_1 \\
t & p_1 < \frac{c}{u} \le p_2 \\
0 &  \frac{c}{u} > p_2 
\end{cases},
$$
and we say that $f$ is characterized by $(t, p_1, p_2)$.
\end{observation}
Observation~\ref{obs:two_step} allows us to think of the allocation rule of \textsc{Greedy} as a lottery (distribution) of at most two posted prices:
\begin{observation}\label{obs:two_prices}
Given an allocation rule $f$ of \textsc{Greedy} that is characterized by $(t, p_1, p_2)$ where $t\in[0,1)$ and $0^-\le p_1\le p_2$, consider the following randomized posted-price mechanism:
\begin{quote}
    For each seller, the buyer independently tosses a (biased) random coin and offers this seller either (i) a payment-per-utility $p_2$ with probability~$t$; or (ii) a payment-per-utility $p_1$ with probability~$1-t$. Then, each seller can accept the offer (give the item to the buyer and receive the payment) or leave.
\end{quote}

The above randomized posted-price mechanism, which is a lottery of two posted prices, has the same allocation function as $f$ in expectation.
\end{observation}
\begin{proof}
Let $\bar{f}$ denote the expected allocation function of the above randomized posted-price mechanism. Now we show that $\bar{f}$ is equivalent to $f$. First, a seller with a cost-per-utility $\frac{c}{u}\le p_1$ will accept either offer $p_1$ or $p_2$
(because both payments-per-utility are no less than his cost-per-utility), and hence $\bar{f}(\frac{c}{u})=1$. On the other hand, a seller with a cost-per-utility $\frac{c}{u}\in (p_1, p_2]$ will only accept offer $p_2$ (because only $p_2$ is no less than his cost-per-utility), and hence $\bar{f}(\frac{c}{u})=\Pr[p_2\textrm{ is offered}]=t$. Finally, a seller with a cost-per-utility $\frac{c}{u}>p_2$ will not accept either of the offer (because both payments-per-utility are below his cost-per-utility), and hence $\bar{f}(\frac{c}{u})=0$.
\end{proof}
The same allocation function obviously achieves the same total utility in expectation, and moreover, by Myerson's lemma, it also makes the same total payment in expectation. Therefore, Observation~\ref{obs:two_prices} provides a more intuitive way to calculate the total utility and the total payment for \textsc{Greedy} (using the posted prices rather than explicitly using the allocation rule of \textsc{Greedy} and Myerson's payment rule), which we formalize in the following observation:
\begin{observation}\label{obs:utility_and_payment}
Given an allocation rule $f$ of \textsc{Greedy} that is characterized by $(t, p_1, p_2)$ where $t\in[0,1)$ and $0^-\le p_1\le p_2$, for any subset of sellers $S\subseteq[n]$, let $U_f(S)$ and $B_f(S)$ denote the total utility and the total payment respectively when we apply $f$ to the sellers in $S$, and let $U_p(S)$ and $B_p(S)$ denote the total utility and total payment respectively when we offer a posted price (payment-per-utility) $p\in\mathbb{R}_{\ge0}$ to the sellers in $S$ (and each seller can accept the offer or leave). Then, we have that
\begin{align}
    &U_f(S)=(1-t)U_{p_1}(S)+tU_{p_2}(S),\nonumber\\
    &B_f(S)=(1-t)B_{p_1}(S)+tB_{p_2}(S),\label{eq:f_to_ps}
\end{align}
and moreover, for all $p\in\mathbb{R}_{\ge0}$,
\begin{align}
    &B_{p}(S)=pU_{p}(S),\nonumber\\
    &U_{p}(S)=\sum_{i\in S\textnormal{ s.t. } c_i/u_i\le p} u_i.\label{eq:B_and_U_for_p}
\end{align}
\end{observation}
\begin{proof}
Eq.~\eqref{eq:f_to_ps} follows immediately by Observation~\ref{obs:two_prices} and the discussion above, and Eq.~\eqref{eq:B_and_U_for_p} follows by definition of the posted-price mechanism.
\end{proof}
We remark that Observation~\ref{obs:utility_and_payment} makes it easier to prove multiplicative concentration inequalities for $U_f(S)$ and $B_f(S)$ when $S$ is a random subset of $[n]$ (specifically, by Eq.~\eqref{eq:f_to_ps} and Eq.~\eqref{eq:B_and_U_for_p}, both $U_f(S)$ and $B_f(S)$ can be written as non-negative linear combination of $U_{p_1}(S)$ and $U_{p_2}(S)$, and thus, it suffices to prove multiplicative concentration inequalities for $U_{p_1}(S)$ and $U_{p_2}(S)$).

\subsection{Approximating greedy via random sampling}
We have shown that \textsc{Greedy} is instance-optimal compared to all the uniform mechanisms in Theorem~\ref{thm:greedy_instance_optimal}, but it requires the knowledge of private costs. In this subsection, we present a proxy of \textsc{Greedy} called \textsc{Random-Sampling-Greedy}, which uses random sampling\footnote{An alternative method often used in the literature is for every seller, computing the prices for the market excluding this seller and then offering the computed prices to this seller. We remark there exist instances for which this method violates budget-feasibility when applied to our idealized mechanism. Besides, random partitioning is much more computationally efficient than this alternative method.} to approximate the distribution of private costs, and in Theorem~\ref{thm:random_sample_greedy}, we will show that this randomized mechanism strictly satisfies the budget constraint and with high probability achieves almost the same utility as \textsc{Greedy}. 

Before that, we introduce two subroutines that will be applied in \textsc{Random-Sampling-Greedy}.
The first subroutine handles an edge case of a small subset $T$ of sellers with exceptionally high utility.
The second subroutine adjusts the price $p_1$ to a new price $\hat{p_1}$, to handle an edge case where $U_{p_1}([n]\setminus T)$ is very small. 
Intuitively, after those adjustments , the utility of any individual seller, who is not in $T$ and has a cost-per-utility at most $\hat{p_1}$, is tiny relative to $U_{\hat{p_1}}([n]\setminus T)$. Therefore when $S$ is a uniformly random subset of $[n]\setminus T$, $U_{\hat{p_1}}(S)$ is concentrated around its expectation w.h.p. (We will show this formally in the analysis of \textsc{Random-Sampling-Greedy}.)

\paragraph{Pre-purchasing the most valuable items} The first subroutine, which will be the first step of \textsc{Random-Sampling-Greedy}, is pre-purchasing the items of highest utilities. By the small-bidder assumption, each seller's cost is $o(B)$. Thus, for an arbitrarily large integer constant $C$, we can pre-purchase the top $C$ items of highest utilities by making a payment $\epsilon_1 B/C$ to each of the $C$ sellers, and the remaining budget is $(1-\epsilon_1)B$. Henceforth, we let $T$ denote the set of the top-$C$ items and let $U(T)$ denote their total utility.

\paragraph{Truncating a greedy allocation rule} We let $\eta>0$ be a parameter which we use for this truncation step (later we will choose $\eta$ to be an arbitrarily small constant and then choose $C$ such that $\eta C$ is arbitrarily large). Suppose we are given an allocation rule $f$ of \textsc{Greedy} that is characterized by $(t,p_1,p_2)$ where $t\in[0,1)$ and $0^-\le p_1\le p_2$. We let $\hat{f}$ denote the truncated allocation rule of $f$. Specifically, $\hat{f}$ is characterized by $(t,\hat{p_1},p_2)$, and $\hat{p_1}$ is defined as follows
\[
    \hat{p_1}:=\begin{cases}
        0^- & U_{p_1}([n]\setminus T)< \frac{U(T)}{\eta C} \\
        p_1 & U_{p_1}([n]\setminus T)\ge \frac{U(T)}{\eta C}
    \end{cases}.
\]
That is, we get $\hat{f}$ by decreasing the value of $f$ over $[0,p_1]$ to $t$ if $U_{p_1}([n]\setminus T)$ is less than $\frac{U(T)}{\eta C}$ (recall that $U_{p_1}([n]\setminus T)$  is the total utility of the sellers in $[n]\setminus T$ whose cost-per-utility is at most $p_1$ by Observation~\ref{obs:utility_and_payment}).

We observe that applying truncation will not significantly decrease (relative to $U(T)$) the utility attained by the allocation rule:
\begin{observation}\label{obs:truncation_does_not_hurt}
For all $S\subseteq[n]$, $U_{f}(S)-U_{\hat{f}}(S)\le \frac{U(T)}{\eta C}$.
\end{observation}
\begin{proof}
By our design of the truncation step and Observation~\ref{obs:utility_and_payment}, $U_{f}(S)-U_{\hat{f}}(S)=(1-t)(U_{p_1}(S)-U_{\hat{p_1}}(S))\le U_{p_1}(S)-U_{\hat{p_1}}(S)$. Moreover, by definition of $\hat{p_1}$, $U_{p_1}(S)-U_{\hat{p_1}}(S)=0$ if $U_{p_1}(S)\ge \frac{U(T)}{\eta C}$, and obviously $U_{p_1}(S)-U_{\hat{p_1}}(S)\le U_{p_1}(S)< \frac{U(T)}{\eta C}$ if otherwise.
\end{proof}

\subsubsection*{The \textsc{Random-Sampling-Greedy} mechanism}

Now we present \textsc{Random-Sampling-Greedy} (Mechanism~\ref{mech:RS_greedy}) and its theoretical guarantee (Theorem~\ref{thm:random_sample_greedy}). The analysis of \textsc{Random-Sampling-Greedy} (the proof of Theorem~\ref{thm:random_sample_greedy}), which is rather technical but still interesting, is deferred to Section~\ref{section:proof_random_sample_greedy} in appendix for the interest of space.

\begin{algorithm}[ht]
\SetAlgoLined
\SetKwInOut{Input}{Input}
\SetKwInOut{Output}{Output}
\Input{$(c_i,u_i)$ for $i\in [n]$, $B$, and parameters $\epsilon_1$, $\delta_1$, $\eta$, $C$.}
\SetAlgorithmName{Mechanism}~~
 Buy the items from the top $C$ sellers $T$ of highest utilities and pay each of them $\epsilon_1 B/C$\;\label{algline:topseller}
 Partition the other sellers $[n]\setminus T$ into $X$ and $Y$ uniformly at random\;
 (Virtually, aka without making actual allocations or payments) run $\textsc{Greedy}$ mechanism on $X$ and $Y$ with budget $\frac{(1-\delta_1)B}{2}$, separately, and get the resulting allocation rules $f_X, f_Y$\;
 Truncate $f_X, f_Y$ using parameter $\eta$ and get $\hat{f_X}, \hat{f_Y}$ and their associated payment rules $\hat{Q_X}, \hat{Q_Y}$\;
 In an arbitrary order, sequentially apply $\hat{f_X},\hat{Q_X}$ to the sellers in $Y$ until we spend $\frac{(1-\epsilon_1)B}{2}$ on $Y$, and then sequentially apply $\hat{f_Y},\hat{Q_Y}$ to the sellers in $X$ until we spend $\frac{(1-\epsilon_1)B}{2}$ on $X$\;\label{algline:otherseller}
 \caption{\textsc{Random-Sampling-Greedy}}
 \label{mech:RS_greedy}
\end{algorithm}

\begin{theorem}\label{thm:random_sample_greedy}
For divisible items, under the small-bidder assumption, for every $\epsilon>0$, there exists sufficiently small $\delta_1,\eta>0$ and sufficiently large $C$ such that, \textsc{Random-Sampling-Greedy} is truthful-in-expectation and strictly budget-feasible and with high probability achieves utility at least $(1-\epsilon)$-fraction of the utility attained by \textsc{Greedy}.
\end{theorem}

\subsection{Numerical simulation on synthetic instances}\label{section:sythetic}
We compare the performance of \textsc{Greedy}, \textsc{Random-Sampling-Greedy} with~\cite{anari2018budget}'s mechanism \textsc{AGN}, and the best cutoff rule with proper tie breaking \textsc{Cutoff} (i.e.,~\cite{EG14}'s open clock auction) on synthetic datasets, where the market has 1000 sellers, each of whom has unit utility and cost sampled from various distributions (negative cost is rounded to $0$), and the buyer's budget is 20000. When we run \textsc{Random-Sampling-Greedy} for this instance, we simply set $\epsilon_1,\delta_1,\eta,C$ to $0$ (these constants were only used to prove asymptotically high probability bounds).
The results are summarized in Table~\ref{table:synthetic} (for each cost distribution, we take the average and the standard deviation of the results of 100 runs).
We observe that \textsc{Greedy} always dominates other mechanisms since it is instance-optimal uniform mechanism, and \textsc{Random-Sampling-Greedy} (\textsc{RS-Greedy}) is usually almost as good as \textsc{Greedy} with only a small difference due to random sampling (as illustrated in Figure~\ref{fig:random_error}, this difference goes to $0$ when the size of market increases, which matches Theorem~\ref{thm:random_sample_greedy}).
Moreover, on all the synthetic instances, \textsc{Greedy} and \textsc{RS-Greedy} beat the worst-case $1-1/e$ competitive ratio by a large margin, while \textsc{AGN} only obtains close-to-$(1-1/e)$ competitive ratio.
On the other hand, for unimodal distributions, \textsc{Cutoff} often performs well, but for multi-modal distributions, it is significantly outperformed by \textsc{Greedy} and \textsc{RS-Greedy}. This matches our intuition:
\begin{example}
\normalfont
Consider the instance where all $n$ sellers have unit utilities, and $n/2$ sellers have costs $0$ and $n/2$ sellers have costs $1$, and the buyer has budget $n$. The best cutoff rule (i.e., setting a best uniform price-per-utility for all sellers) only gets the $n/2$ sellers with zero cost and hence achieves competitive ratio $1/2$, while \textsc{Greedy} can choose tuple $(t=1/2,\,p_1=0,\,p_2=1)$ and achieve competitive ratio $3/4$.
\end{example}

\begin{table}[h!]
\begin{minipage}{\columnwidth}
\begin{center}
\caption{\label{table:synthetic} Competitive ratios achieved by different mechanisms on synthetic datasets}
\begin{tabular}{ c c c c } 
 \hline
 \textsc{Cutoff} & \textsc{AGN} & \textsc{Greedy} & \textsc{RS-Greedy} \\
\hline
 $0.816\pm 0.004$ & $0.632\pm 0.001$ & $0.818\pm 0.004$ & $0.81\pm 0.006$ \\

 $0.709\pm 0.005$ & $0.633\pm 0.003$ & $0.711\pm 0.004$ & $0.702\pm 0.006$ \\

 $0.74\pm 0.008$ & $0.663\pm 0.006$ & $0.743\pm 0.008$ & $0.736\pm 0.009$ \\

 $0.69\pm 0.003$ & $0.633\pm 0.002$ & $0.726\pm 0.003$ & $0.718\pm 0.005$ \\

 $0.68\pm 0.009$ & $0.634\pm 0.003$ & $0.712\pm 0.006$ & $0.706\pm 0.007$  \\
 \hline
\end{tabular}
\end{center}
\footnotesize
Each row contains the results for a distinct cost distribution. The distributions from top to bottom are $\mathcal{N}(20,5)$, $\textrm{Unif}(0,40)$, $\textrm{Exp}(20)$, $\frac{1}{2}\mathcal{N}(10,3)+\frac{1}{2}\mathcal{N}(30,3)$, $\frac{1}{3}\mathcal{N}(5,3)+\frac{1}{3}\mathcal{N}(20,3)+\frac{1}{3}\mathcal{N}(35,3)$.
\end{minipage}
\end{table}

\begin{figure}[h!]
\begin{center}
    \begin{subfigure}[b]{0.32\textwidth}
        \centering
        \includegraphics[scale=0.37]{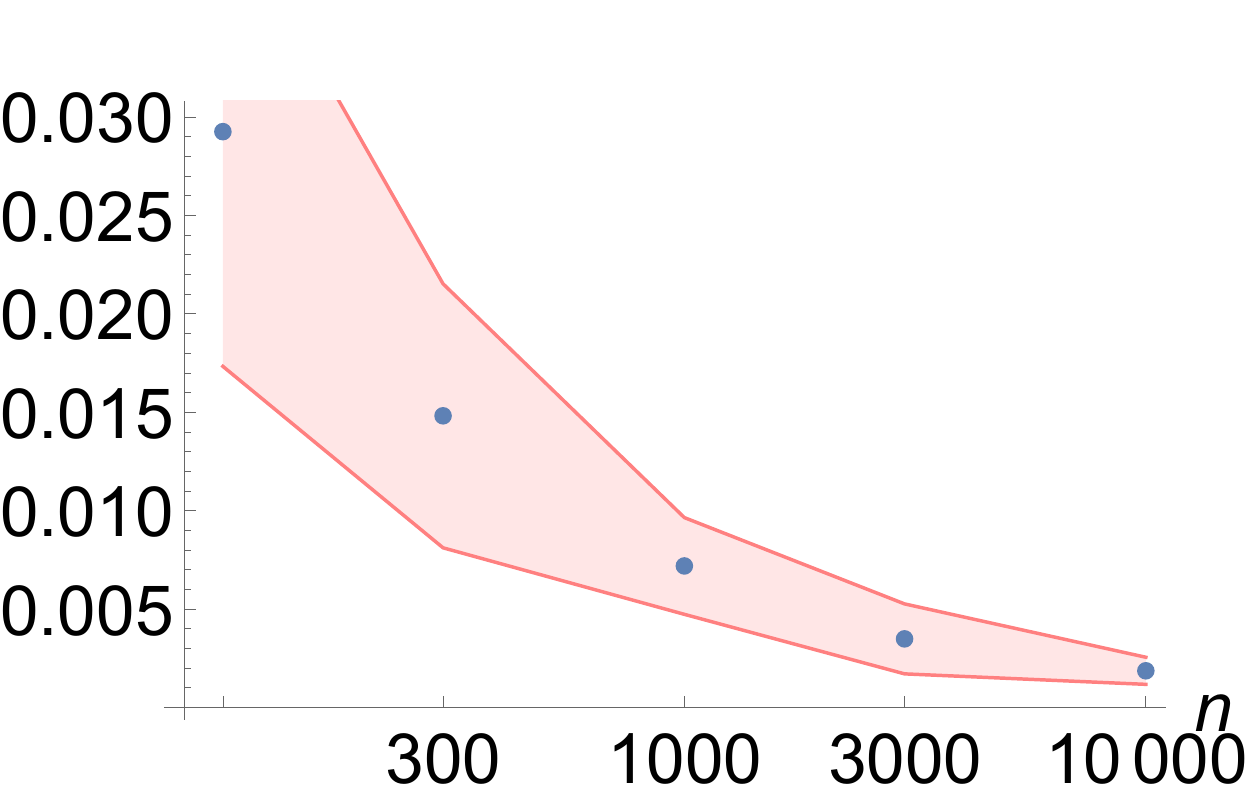}
        \caption{$\mathcal{N}(20,5)$}
    \end{subfigure}%
    ~
    \begin{subfigure}[b]{0.32\textwidth}
        \centering
        \includegraphics[scale=0.37]{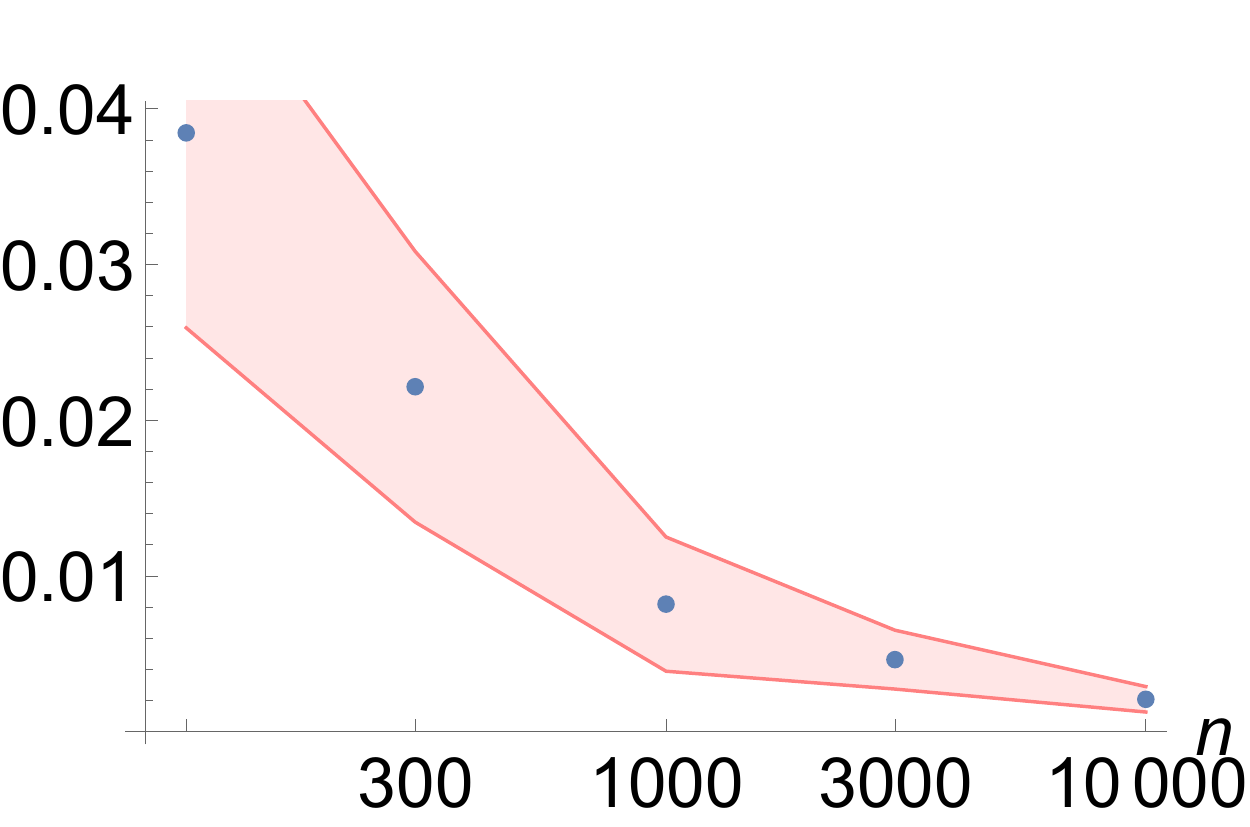}
        \caption{$\textrm{Unif}(0,40)$}
    \end{subfigure}%
    ~
    \begin{subfigure}[b]{0.32\textwidth}
        \centering
        \includegraphics[scale=0.37]{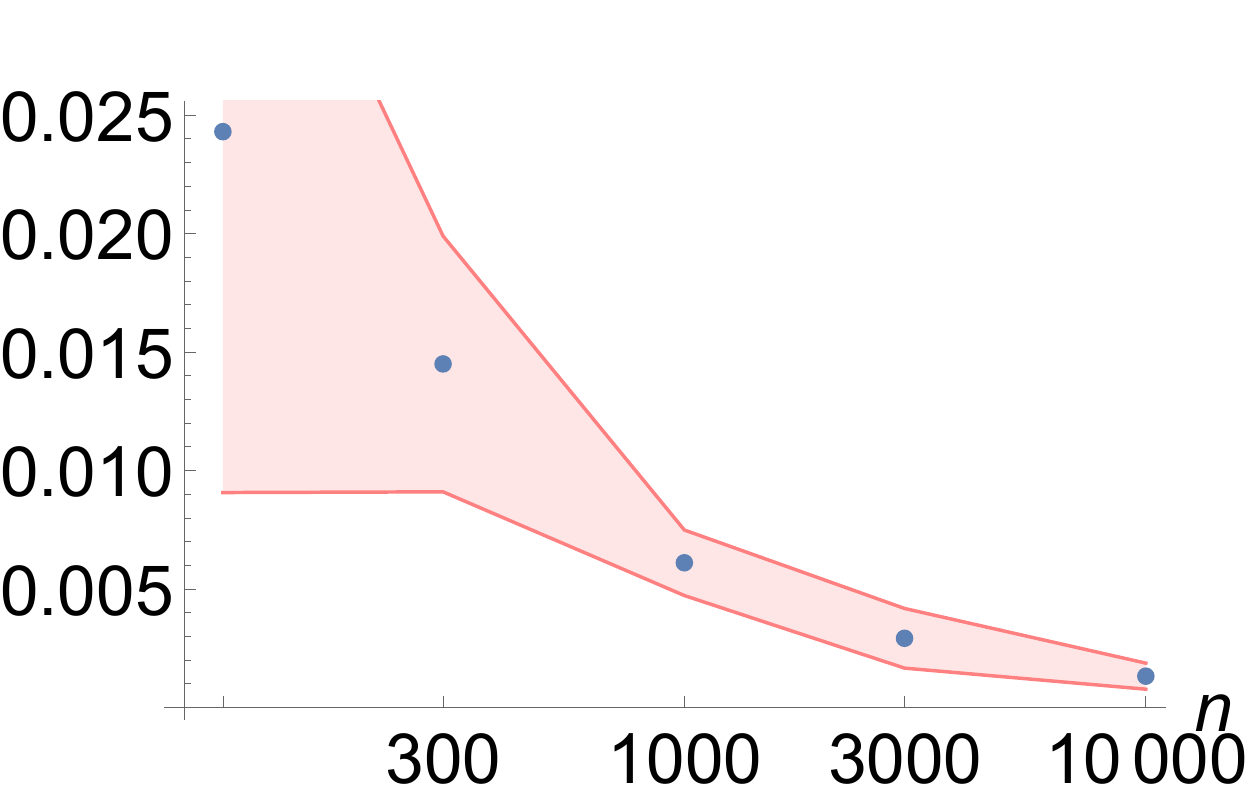}
        \caption{$\textrm{Exp}(20)$}
    \end{subfigure}
\end{center}
  \caption{This figure shows that the difference between competitive ratios achieved by \textsc{Greedy} and \textsc{RS-Greedy} ($y$-axis) diminishes when the market size $n$ ($x$-axis) increases. In each subplot, the market has $n$ sellers, each of whom has unit utility and cost sampled from a distinct distribution (negative cost is rounded to $0$), and the buyer's budget is $20n$. Each datapoint in the plot is the average of 20 runs, and the shaded area captures one standard deviation.}
\label{fig:random_error}
\end{figure}

\section{Budget-smoothed analysis}\label{sec:smoothed}
In this section, we analyze our mechanism in the budget-smoothed analysis framework. Our main results of budget-smoothed analysis are:
\begin{itemize}
\item Our mechanism obtains near-optimal budget-smoothed competitive ratio for any budget distribution when compared to all (possibly non-uniform) mechanisms (Theorem~\ref{thm:worst_distribution}).
\item Our mechanism obtains strictly better than $1-1/e$ budget-smoothed competitive ratio on any non-trivial budget distribution (Theorem~\ref{thm:beating}). In Section~\ref{section:numerical_computing_ratios}, we also formulate a (non-convex) mathematical program that computes the budget-smoothed competitive ratio for any given budget distribution. We solve this program for various distributions and observe non-negligible improvement over $1-1/e$.
\item Given any bounded range of budgets, there is a single market on which, simultaneously for every budget in the range,~\cite{anari2018budget}'s mechanism obtains only $1-1/e$ competitive ratio (Theorem~\ref{thm:hard_instance_for_agn}).
\item Our mechanism (and hence all possibly non-uniform mechanisms by Theorem~\ref{thm:worst_distribution}) has budget-smoothed competitive ratio bounded away from $1$ (specifically, at most $0.854$) for any budget distribution (Theorem~\ref{thm:lower_bound}).
\end{itemize}

\subsection{Greedy is optimal for any budget distribution}~\label{section:greedy_budget_smooth}
In this subsection, we analyze the budget-smoothed competitive ratio of our complete-information ``mechanism'' \textsc{Greedy} for any budget distribution. We show that \textsc{Greedy} is optimal for any budget distribution --even among non-uniform mechanisms-- and the ratio goes beyond $1-1/e$ when there are multiple budgets in the support of the budget distribution. These results extend to \textsc{Random-Sampling-Greedy} due to Theorem~\ref{thm:random_sample_greedy}. We will characterize the worst Bayesian market for truthful-in-expectation uniform mechanisms, where $n$ sellers have the same utility, and their costs are drawn from a continuous distribution. This characterization can be viewed as a generalization of the worst-case instance\footnote{J.Z.~wants to thank Nima Anari for an inspiring discussion of the worst-case instance in~\cite{anari2018budget}.} in~\cite{anari2018budget}. Then, we argue that this Bayesian market is as hard for truthful-in-expectation non-uniform mechanisms. But before that, we explain why the continuous cost distribution and equal utilities are not restrictions, \ie, for an arbitrary market, we can construct a Bayesian market with continuous cost distribution and equal utilities that exhibits the same hardness for truthful-in-expectation uniform mechanisms as the original market.

\subsubsection{From arbitrary market to Bayesian market}
Given an market $I$ with $n$ sellers of utilities $u_i$'s and costs $c_i$'s, we first construct a Bayesian market $I_1$ with a discrete distribution. The market $I_1$ has $M\cdot\sum_i u_i$ sellers\footnote{Without loss of generality, we assume that $M\cdot\sum_i u_i$ is an integer.}, where $M$ is a sufficiently large number. Let $D_1$ be a distribution over $\{\frac{c_i}{M u_i}\mid i\in[n]\}$ such that the probability of $\frac{c_i}{M u_i}$ is $u_i/(\sum_j u_j)$. Each seller has utility $\frac{1}{M}$, and his cost is drawn from $D_1$. We need to verify two things: (i) the non-IC optimal utilities of the knapsacks for $I$ and $I_1$ are almost equal, and (ii) the best achievable utilities by uniform mechanisms for $I$ and $I_1$ are also almost the same. For (ii), it suffices to consider \textsc{Greedy} because of Theorem~\ref{thm:greedy_instance_optimal}. The reason both of these hold is that the optimal utility and the best achievable utility only depend on the cost-to-utility ratio $\frac{c}{u}$'s and the total utility of the sellers with the same $\frac{c}{u}$, and if $M$ is sufficiently large, with high probability these quantities do not change much in $I_1$ compared to $I$.

Next, we construct a Bayesian market $I_2$ with the same setup as $I_1$ but a continuous distribution for sellers' costs. To this end, consider the CDF of $D_1$, which is some step function $F(c)$, we can approximate each step in $F$ arbitrarily well by a logistic function and glue them together such that the CDF is differentiable. For the same reason as above, the best achievable competitive ratio of a uniform mechanism for $I_2$ is approximately equal to that for $I_1$.

\subsubsection{Characterizing the worst Bayesian market}
\begin{theorem}\label{thm:worst_distribution}
For any distribution $\mathcal{D}$ over any $m$ budgets $B_1<B_2\dots<B_m$, let $F(c)$ be the CDF of the distribution of costs of the worst\footnote{By ``worst for $\mathcal{D}$'', we mean it minimizes the best possible expected competitive ratio that is achievable by any mechanism given budget distribution $\mathcal{D}$.} Bayesian market for $\mathcal{D}$. Then, the following hold:
\begin{itemize}
\item Consider the plot of $cF(c)$ with respect to $F(c)$. $cF(c)$  is a piecewise-linear function of $F(c)$ with at most $m$ non-zero linear pieces, and has non-decreasing slope. 
\item For each budget, the utility-maximizing allocation rule for this market is a uniform cutoff rule, namely, $f(c/u)=\mathds{1}(c/u\le c^*/u^*)$ for some $c^*/u^*$.
\item \textsc{Greedy} is worst-case optimal for budget distribution $\mathcal{D}$ (see Defintion~\ref{def:budget_smoothed_optimal}) compared to all the truthful-in-expectation (not necessarily uniform) mechanisms. (Note that the optimality also holds for \textsc{Random-Sampling-Greedy} due to Theorem~\ref{thm:random_sample_greedy}.)
\end{itemize}
\end{theorem}
\begin{proof}
From the previous discussion, it suffices to consider a Bayesian market, where $n$ sellers have the same utility, and their costs are drawn from a continuous distribution, the CDF of which is some continuous $F$. The following min-max program computes the cost distribution that gives the worst expected competitive ratio for budget distribution $\mathcal{D}$ against uniform allocation rules (later we will show that non-uniform rules are not any better for the worst distribution),
\begin{equation*}
\begin{split}
    \inf_{F}&\sup_{f_1,\dots,f_m} \sum_{i=1}^m \Pr[B_i] \cdot \frac{F(0)+\int_{0^+}^{\infty} f_i(c)\,dF(c)}{\tau_i} \\
    \textrm{s.t.}\quad &\forall i\in[m],\, \underbrace{\int_{0}^{\tau_i} c\,dF(c)}_{\textrm{total payment of non-IC optimum}}=\underbrace{Q_{f_i}(0)\cdot F(0)+\int_{0^+}^{\infty} Q_{f_i}(c)\,dF(c)}_{\textrm{Myerson's payment for $f_i$}}=B_i,\\
    & \forall i\in[m],\forall c\ge 0,\, f_i(c)\in[0,1],\\
    &\textrm{and } F \textrm{ is a continuous CDF},
\end{split}
\end{equation*}
where $\Pr[B_i]$ is the probability of $B_i$ according to $\mathcal{D}$, $f_i$ is the allocation function for $i$-th budget, and $\tau_i$'s denote the expected non-IC optimal utility for the corresponding budgets, and we hard-code in the program those $\tau_i$'s which result in the worst expected ratio. Although we did not require $f_i$'s to be monotone here, later we will show that if we add the monotonicity constraint, the worst ratio does not change. Also, note that we only require budget feasibility in expectation for the allocation function, and hence, the optimality of \textsc{Greedy} will hold even among ex ante budget-feasible mechanisms. We should have restricted the non-IC optimal solution to be ex post budget-feasible, but this is fine, because as market size $n$ grows, with high probability, the budget spent in the optimal solution is concentrated around its expectation, and Lemma~\ref{lem:knapsack_concave} says cutting the budget slightly does not decrease the optimal utility much.

Now we derive that 
\begin{align*}
\int_{0^+}^{\infty} &Q_f(c)\,dF(c) =\int_{0^+}^{\infty} \left(f(c)\cdot c+ \int_{c}^{\infty}f(x) \,dx \right) \,dF(c)  && \text{(Myerson's payment identity)}\\
&=\int_{0^+}^{\infty} f(c)\cdot c\, dF(c) + \left(F(c)\cdot\int_{c}^{\infty}f(x) \,dx\right)|_{0^+}^\infty \\
&\qquad-\int_{0^+}^{\infty} F(c)\,d\left(\int_{c}^{\infty}f(x) \,dx\right) && \text{(Integration by parts)}\\
&=\int_{0^+}^{\infty} f(c)\cdot c\, dF(c) -F(0)\cdot\int_{0^+}^{\infty}f(x) \,dx \\
&\qquad + \int_{0^+}^{\infty} f(c)\underbrace{F(c)dc}_{=\frac{F(c)}{F'(c)} dF(c)}\\
&=\int_{0^+}^{\infty} f(c)\cdot\left(c+\frac{F(c)}{F'(c)}\right)\,dF(c)-F(0)\cdot Q_f(0).    
\end{align*}
Therefore, the maximization problem in the min-max program can be seen as a fractional knapsack (where $dF(c)$ is the value of an item $c$, and $c+\frac{F(c)}{F'(c)}$ is its weight per value), and the best allocation function should choose the $c$'s with small $c+\frac{F(c)}{F'(c)}=\frac{d(cF(c))}{dF(c)}$.
We now prove several structural properties about the plot (curve) of $cF(c)$ with respect to $F(c)$. 

\paragraph{Any feasible curve should have non-decreasing slope from the origin}
Notice that $\frac{d(cF(c))}{dF(c)}$ is the slope of this curve at $F(c)$, and $c$ is the slope of the line from the origin to the point of the curve at $F(c)$. Any feasible curve should have non-decreasing slope from the origin since $F(c)$ is non-decreasing in $c$ and vice versa. 

Now we show the structural result about the worst-case $F$ for one budget, and later we will extend to many budgets. 

\paragraph{The curve is piecewise-linear w.l.o.g.}
Given a feasible curve for some $F$, we discretize the smooth curve into a piecewise-linear curve. The discretization is sufficiently fine-grained such that the slope of the curve and the slope to the origin at each point are close to those of the original curve, and there are only finitely many non-differentiable points. Hence the min-max program is still valid, and its result does not change much. It suffices to consider such piecewise-linear curves.

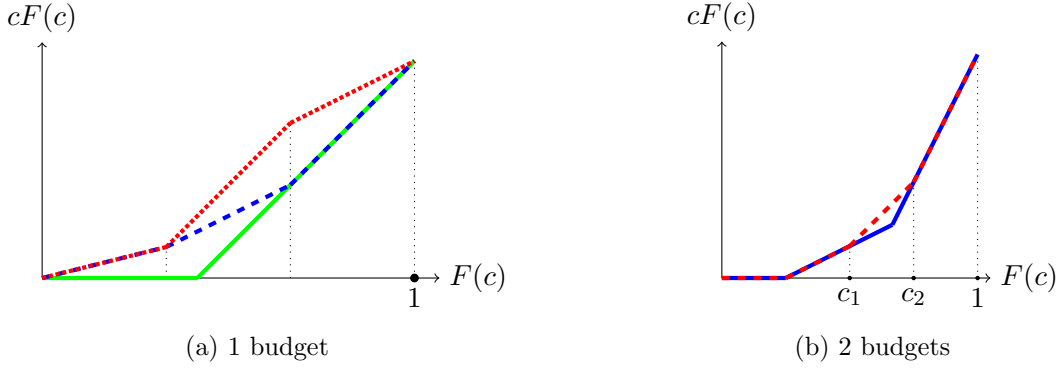
\begin{figure}[t!]
    \begin{subfigure}[b]{0.48\textwidth}
        \centering
   \begin{tikzpicture}[scale=3.3]
      \draw[->] (0,0) -- (1.6,0) node[right] {$F(c)$};
      \draw[->] (0,0) -- (0,0.95) node[above] {$cF(c)$};
      \draw[dotted] (1.5,0) -- (1.5,0.5*1.75);
      \draw[dotted] (1,0) -- (1,1.25*0.5);
      \draw[dotted] (0.5,0) -- (0.5,0.5*0.25);
      \fill [black] ($(1.5,0)$) circle (0.5pt) node at (1.5,0)[below]{\color{black}$1$};
      
      \draw[scale=1,domain=0:1-0.25*1.5,smooth,variable=\x,green,ultra thick] plot ({\x},{0});
      \draw[scale=1,domain=1-0.25*1.5:1.5,smooth,variable=\x,green,ultra thick] plot ({\x},{\x-1+0.25*1.5});
      
      \draw[dashed,scale=1,domain=0:0.5,smooth,variable=\x,blue,ultra thick] plot ({\x},{0.25*\x});
      \draw[dashed,scale=1,domain=0.5:1,smooth,variable=\x,blue,ultra thick] plot ({\x},{0.5*(\x-0.5)+0.25*0.5});
      \draw[dashed,scale=1,domain=1:1.5,smooth,variable=\x,blue,ultra thick] plot ({\x},{\x-1+0.25*1.5});
      
      \draw[densely dotted,scale=1,domain=0:0.5,smooth,variable=\x,red,ultra thick] plot ({\x},{0.25*\x});
      \draw[densely dotted,scale=1,domain=0.5:1,smooth,variable=\x,red,ultra thick] plot ({\x},{\x-0.5+0.25*0.5});
      \draw[densely dotted,scale=1,domain=1:1.5,smooth,variable=\x,red,ultra thick] plot ({\x},{0.5*(\x-1)+1.25*0.5});
\end{tikzpicture}
        \caption{1 budget\label{fig:spoon}}
    \end{subfigure}%
    ~ 
    \begin{subfigure}[b]{0.48\textwidth}
        \centering
       \begin{tikzpicture}[scale=1.7]
      \draw[->] (0,0) -- (2.1,0) node[right] {$F(c)$};
      \draw[->] (0,0) -- (0,1.85) node[above] {$cF(c)$};
      \draw[dotted] (2,0) -- (2,1.75);
      \draw[dotted] (1.5,0) -- (1.5,0.75);
      \draw[dotted] (1,0) -- (1,0.25);
      \fill [black] ($(2,0)$) circle (0.5pt) node at (2,0)[below]{\color{black}$1$};
      \fill [black] ($(1.5,0)$) circle (0.5pt) node at (1.5,0)[below]{\color{black}$c_2$};
      \fill [black] ($(1,0)$) circle (0.5pt) node at (1,0)[below]{\color{black}$c_1$};

      \draw[solid,scale=1,domain=0:0.5,smooth,variable=\x,blue,ultra thick] plot ({\x},{0});
      \draw[solid,scale=1,domain=0.5:4/3,smooth,variable=\x,blue,ultra thick] plot ({\x},{0.5*(\x-0.5)});
      \draw[solid,scale=1,domain=4/3:2,smooth,variable=\x,blue,ultra thick] plot ({\x},{2*(\x-1.5)+0.75});
      
      \draw[dashed,scale=1,domain=0:0.5,smooth,variable=\x,red,ultra thick] plot ({\x},{0});
      \draw[dashed,scale=1,domain=0.5:1,smooth,variable=\x,red,ultra thick] plot ({\x},{0.5*(\x-0.5)});
      \draw[dashed,scale=1,domain=1:1.5,smooth,variable=\x,red,ultra thick] plot ({\x},{(\x-1)+0.25});
      \draw[dashed,scale=1,domain=1.5:2,smooth,variable=\x,red,ultra thick] plot ({\x},{2*(\x-1.5)+0.75});
      
\end{tikzpicture}
        \caption{2 budgets\label{fig:spoon_2}}
    \end{subfigure}
    \caption{{\bf On the left:} Starting from an arbitrary piecewise linear curve (red dotted), we can re-order its pieces to get blue dashed curve and then again into the green solid curve. These steps only make the market worse. The worst $F$ (green solid) for one budget is ``ReLU shaped''. \\
{\bf On the right:} The worst $F$ for $2$ ($m$ respectively) budgets should have at most $2$ ($m$ respectively) non-zero linear pieces. Consider the optimal allocation functions for two budgets, which are cutoff rules, if neither cutoff lies in $(c_1,c_2)$, then changing the red dashed curve into the blue solid curve makes the market worse.}
\label{fig:visualize_main_steps}
\end{figure}

\paragraph{Worst-case curve has non-decreasing slope}
Our first observation is that the worst-case curve should have non-decreasing slope. If it does not, we can re-order the linear pieces according to their slopes, and the re-ordering preserves the probability mass of $c$'s with any fixed slope. Hence the best allocation rule makes the same utility as before. Meanwhile, the slope to the origin at each $c$ can only become smaller than that before re-ordering. The budget spent by the non-IC optimal solution to get the same utility as before is the integration of the slope from origin from $0$ to some $\tau$, which can only decrease. Therefore, the re-ordering can only decrease the competitive ratio of the best allocation rule. This is illustrated in Figure~\ref{fig:visualize_main_steps}~(\subref{fig:spoon}), when we re-order the red dotted curve and get the blue dashed curve. 

A claim following from the non-decreasing slope is that the best allocation rule should be a cutoff rule.
\begin{claim}\label{claim:cutoff}
    Utility-maximizing allocation function for a convex $F(c)$-to-$cF(c)$ curve is a cutoff rule.
\end{claim}
\begin{proof}[Proof of Claim~\ref{claim:cutoff}]
Indeed, since the best allocation rule comes from solving the fractional knapsack problem we mentioned above and the value per weight (equal to slope) is non-decreasing, the solution should be $f(c)=1$ for all $c\le c_1$, and $f(c)=t<1$ for all $c_1<c\le c_2$, for some $c_1<c_2$. This rule can be seen as a probabilistic combination of two cutoff rules, \ie, with some probability $\alpha$ offer cutoff price $c_1$ and offer $c_2$ otherwise, and the expected utility and payment are $\alpha F(c_1)+(1-\alpha)F(c_2)$ and $\alpha c_1F(c_1)+(1-\alpha)c_2F(c_2)$ respectively. Consider the $c_3$ such that $F(c_3)=\alpha F(c_1)+(1-\alpha)F(c_2)$, because the curve is convex, $c_3F(c_3)\le \alpha c_1F(c_1)+(1-\alpha)c_2F(c_2)$. Hence the cutoff rule at $c_3$ makes the same utility but spends no more than the probabilistic rule, and the claim follows. 
\end{proof}

\paragraph{Worst-case curve for one budget is a ReLU function}
Next, we argue that the worst-case (convex) curve for one budget should be a ReLU function, \ie, it is zero at first and then becomes a linear function. Suppose otherwise, we let $c^*$ be the cutoff price of the best allocation rule. We can draw a line between $(F(c^*),c^*F(c^*))$ and the $F(c)$-axis with slope equal to the slope of the worst curve at $(F(c^*),c^*F(c^*))$. Consider the ReLU curve whose non-zero linear part is this line. Notice that $(F(c^*),c^*F(c^*))$ does not change and is still optimal, and hence the optimal utility achievable by any allocation rule does not change. Meanwhile, the slope from origin at each $c$ can only become smaller, and therefore, the budget spent by the non-IC optimal solution to get the same utility as before can only decrease. This is illustrated in Figure~\ref{fig:visualize_main_steps}~(\subref{fig:spoon}), where we change the blue dashed curve into the green solid curve. Furthermore, the part of the original curve after $(F(c^*),c^*F(c^*))$ has slope larger than the slope at this point, and decreasing this part to the line with the slope at this point only decreases the spent budget for the non-IC optimal solution and does not change the result of the best allocation rule. 
The final curve is a ReLU, and the best achievable competitive ratio only gets worse.

\paragraph{Characterizing the worst-case curve for many budgets}
Now we show that the worst-case curve for $m$ budgets has at most $m$ non-zero linear pieces by generalizing the above argument. Consider the $m$ optimal cutoff rules for $m$ budgets respectively, if there are more than $m$ non-zero linear pieces, then there is one piece from some $(F(c_1),c_{1}F(c_{1}))$ to some other $(F(c_{2}),c_{2}F(c_{2}))$ such that the open interval $(c_1,c_2)$ does not contain any optimal cutoff. If this is not the last piece of the curve, we can extend the piece before $(F(c_1),c_{1}F(c_{1}))$ upwards and the piece after $(F(c_{2}),c_{2}F(c_{2}))$ downwards until they intersect, which decreases the number of linear pieces. Similar to the one budget case, this step does not change the payments of optimal cutoff rules and can only decrease the payments of non-IC optimal solutions, and therefore, this only makes the market worse. This is illustrated in Figure~\ref{fig:visualize_main_steps}~(\subref{fig:spoon_2}), where we change the red curve to the blue. If it is the last piece of the curve, namely $F(c_2)=1$, then we can simply extend the piece before $(F(c_1),c_{1}F(c_{1}))$ upwards until it hits $1$ horizontally. The argument in this case is analogous.

\paragraph{Adding monotonicity constraints to the min-max program}
Next, we explain why restricting $f_i$'s to be monotone does not change the optimal value to the min-max program. As we argued above, the best allocation rules for the worst distribution $F^*$ are cutoff rules $f_i^*$, which are monotone. Since $(F^*,\{f_i^*\mid i\in[m]\})$ is an equilibrium of the min-max program without monotonicity constraints, $(F^*,\{f_i^*\mid i\in[m]\})$ is obviously also an equilibrium of the min-max program with monotonicity constraints. Notice that the min-max program with monotonicity constraints satisfies the conditions of Lemma~\ref{lem:min-max}. Hence the optimal value to this program is equal to the objective value at this equilibrium.

\paragraph{Non-uniform allocation rule is not better}
We show that for a Bayesian market that matches our characterization, non-uniform rules do not outperform uniform rules.
Consider a general (possibly non-uniform) mechanism where each seller $i$ has its own allocation rule $A^{(i)}_{c_{-i}}$.
Now let $P^{(i)}_{c_{-i}}(c)=1-A^{(i)}_{c_{-i}}(c)$. An implementation of $A^{(i)}_{c_{-i}}$ is sampling cutoff prices from the distribution whose CDF is $P^{(i)}_{c_{-i}}$  To see this, the probability that the item of price $c_i$ is bought is $1-P^{(i)}_{c_{-i}}(c_i)=A^{(i)}_{c_{-i}}(c_i)$, and the expected payment is 
\begin{align*}
    \int_{c_i}^{c_{\max}} c dP^{(i)}_{c_{-i}}(c)&= cP^{(i)}_{c_{-i}}(c)|_{c_i}^{c_{\max}}-\int_{c_i}^{c_{\max}} P^{(i)}_{c_{-i}}(c)dc \\
    &=c_{\max}-c_iP^{(i)}_{c_{-i}}(c_i)-\int_{c_i}^{c_{\max}} P^{(i)}_{c_{-i}}(c)dc \\
    &=c_{\max}-c_i(1-A^{(i)}_{c_{-i}}(c_i))-\int_{c_i}^{c_{\max}} (1-A^{(i)}_{c_{-i}}(c))dc \\
    &=c_iA^{(i)}_{c_{-i}}(c_i)+\int_{c_i}^{c_{\max}}A^{(i)}_{c_{-i}}(c)dc,
\end{align*}
which is exactly the Myerson payment corresponding to $A^{(i)}_{c_{-i}}$. Since the allocation rule $A^{(i)}_{c_{-i}}$ can be seen as a probabilistic combination of cutoff rules, as we have shown before, by convexity of the $F(c)$-to-$(cF(c))$ curve (see Claim~\ref{claim:cutoff}), there is a cutoff rule $p^{(i)}_{c_{-i}}$ that achieves the same utility but with less or equal payment compared to $A^{(i)}_{c_{-i}}$. Moreover, $p^{(i)}_{c_{-i}}$'s together can be seen as a probabilistic cutoff rule that depends on random variable $c_{-i}$, and hence again by the same argument, there is a cutoff rule $p_i$ that does as good as the random $p^{(i)}_{c_{-i}}$ for seller $i$. Finally, for the same reason, the uniform cutoff rule $p$ such that $F(p)=\frac{1}{n}\sum_{i=1}^n F(p_i)$ is as good as the non-uniform rule $p_i$'s.

Since \textsc{Greedy} uses the best monotone uniform allocation rule for any instance by Theorem~\ref{thm:greedy_instance_optimal}, the min-max program with additional monotonicity constraint solves for the worst-case expected competitive ratio for \textsc{Greedy}. Thus, the observation in the above paragraph implies that \textsc{Greedy} is worst-case optimal compared to all the truthful-in-expectation (even non-uniform) mechanisms.

\end{proof}

\subsection{Our mechanism beats \texorpdfstring{$1-1/e$}{1-1/e}}

In Section~\ref{section:numerical_computing_ratios}, we will formulate a (non-convex) mathematical program to solve for optimal budget-smoothed competitive ratios for any budget distribution, and we will solve the program numerically for various interesting distributions and observe non-negligible improvement over $1-1/e$. Here we analytically prove that the optimal ratio (i.e., the ratio achieved by the optimal \textsc{Greedy-I}) goes beyond $1-1/e$ when there are multiple budgets in the support of the budget distribution.

\begin{theorem}\label{thm:beating}
    For any budget distribution with support size $\ge 2$, \textsc{Greedy-I} achieves budget-smoothed competitive ratio strictly better than $1-1/e$. (Note that this also holds for \textsc{Random-Sampling-Greedy} due to Theorem~\ref{thm:random_sample_greedy}.)
\end{theorem}

\begin{proof}
We follow the proof of Theorem~\ref{thm:worst_distribution}. In one-budget case, we can solve for the worst-case distribution analytically, which results in the optimal competitive ratio $1-{1}/{e}$. (This corresponds to the worst-case hardness for a fixed budget~\cite{anari2018budget}). Specifically, recall that the $F(c)$-to-$cF(c)$ curve of the worst-case distribution for one budget has one non-zero linear piece. Therefore, there exists $a,b$ such that $cF(c)=aF(c)-b$, which implies that $F(c)=b/(a-c)$. Consider arbitrary $0\le x\le a-b$, the total cost of the sellers, whose individual cost is at most $x$, is
\begin{align*}
    \int_{0}^x cdF(c) &= cF(c)|_0^x - \int_{0}^x F(c)dc \\
    &= xF(x) - \int_{0}^x \frac{b}{a-c} dc \\
    &= \frac{bx}{a-x} - (-b\ln(a-c))|_0^x \\
    &= \frac{bx}{a-x} + b\ln(\frac{a-x}{a}).
\end{align*}
Recall we showed in the proof of Theorem~\ref{thm:worst_distribution} that the best allocation rule for the worst-case distribution for any budget is a cutoff rule. The Myerson's payment of the cutoff rule at cost $y$ is $yF(y)$, and hence, with budget $\frac{bx}{a-x} + b\ln(\frac{a-x}{a})$, the optimal utility is $F(y)$ such that $yF(y)=\frac{bx}{a-x} + b\ln(\frac{a-x}{a})$, and since $yF(y)=aF(y)-b$, it follows that $F(y)=(\frac{ab}{a-x} + b\ln(\frac{a-x}{a}))/a$. Dividing $F(y)$ by $F(x)$, we get the competitive ratio
\[
    \frac{(\frac{ab}{a-x} + b\ln(\frac{a-x}{a}))}{a\frac{b}{a-x}}=1+\frac{a-x}{a}\ln(\frac{a-x}{a}).
\]
Notice that this is a decreasing function on $[0,a-b]$, and hence, the minimum is $1+\frac{b}{a}\ln(\frac{b}{a})$, which is achieved by $x=a-b$. Then, by standard calculus, $1+\frac{b}{a}\ln(\frac{b}{a})$ is minimized by letting $\frac{b}{a}=1/e$, and the minimum is $1-1/e$.

It follows from the monotonicity of the competitive ratio in $x$ that for the worst distribution in one-budget case, the budget for which the optimal competitive ratio is $1-{1}/{e}$ is unique. This implies that when there are $(\ge2)$ budgets, the optimal average competitive ratio is beyond $1-{1}/{e}$. To see this, notice that by Theorem~\ref{thm:worst_distribution}, the worst curve for the two-budget case is a $(\le2)$-piece piecewise-linear function. Based on the observation above, $1$-piece linear function cannot be hard for both budgets. From the proof of Theorem~\ref{thm:worst_distribution}, we know that we can make a $2$-piece piecewise-linear function strictly worse for the larger budget by changing it into a $1$-piece linear function, and hence the ratio is beyond $1-{1}/{e}$ for the larger budget.
\end{proof}

\subsection{The limit of budget-smoothed analysis}~\label{sec:budget-smoothed-lb}
As we mentioned in the introduction, the worst-case hardness result for a fixed budget~\cite{anari2018budget} breaks under a perturbation on the budget constraint. It is tempting to hope that as the distribution of budget perturbation (with respect to an arbitrary budget) becomes arbitrarily spread (a.k.a.~arbitrarily far from the worst-case single budget), the optimal budget-smoothed competitive ratio can become close to $1$.
In this section, we show a negative result, i.e., a strong hardness result that is robust to any budget distribution and any mechanism.

\begin{theorem}\label{thm:lower_bound}
For any distribution of budget perturbations\footnote{A distribution of budget perturbations is a distribution of budgets normalized by a fixed budget.} and any truthful-in-expectation mechanism, there is a market for which the expected competitive ratio of the mechanism is at most $0.854$.
\end{theorem}
\begin{proof}
For arbitrarily small $\epsilon$, suppose that $1-\epsilon$ fraction of the total mass of the distribution of perturbations is on certain $[\rho_{\min},\rho_{\max}]$ (this is w.l.o.g.~even for unbounded distributions).
Since we have shown that \textsc{Greedy} is worst-case optimal compared with all the truthful-in-expectation mechanisms, it suffices to construct a hard market for which \textsc{Greedy} achieves expected competitive ratio $0.854$.
\paragraph{Construction}
Let $q=1+1/\sqrt{2}$ and $w=2$. There are $m+1$ groups of sellers, each having $n$ sellers of the same utility and cost ($m,n$ will be specified shortly). The zeroth group has total utility $u_0=1$ and cost-to-utility ratio $\gamma_0=0$. The $i$-th group has total utility $u_i=w^{i-1}$ and cost-to-utility ratio $\gamma_i=q^{i-1}$. We let $B_i=\sum_{j=1}^i w^{j-1}q^{j-1}$ for $i\in[m]$. The budget distribution is supported on the interval that spans from $B_2$ to $B_m$. Finally, $m$ is chosen such that $B_m/B_2\ge \rho_{\max}/\rho_{\min}$, and $n$ goes to infinity (hence the market satisfies our small-bidder assumption).

\paragraph{}Recall that $e_{i,j} =\frac{\sum_{i\le t\le j}u_t}{(\gamma_j-\gamma_{i-1})\cdot\sum_{0\le t\le i-1}u_t+\gamma_j\cdot\sum_{i\le t\le j}u_t}$ with $j\ge i$ defined in \textsc{Greedy} is the additional utility per additional total payment when we simultaneously increase the fraction bought from sellers $i$ to $j$. We first show that for this market, $e_{i,j}$ is decreasing in $j$ for any $i$, and hence, the mechanism always buys items from the group of the smallest index that has not been exhausted. By definition of $u_t$ and $\gamma_t$ in the construction, we have that for $i\ge 2$ and any $j\ge i$,
\begin{align*}
    \frac{e_{i,j+1}}{e_{i,j}}&=\frac{\sum_{t=i}^{j+1}w^{t-1}}{\sum_{t=i}^{j}w^{t-1}}\cdot\frac{q^{j-1}\cdot\sum_{t=i}^{j}w^{t-1}+(q^{j-1}-q^{i-2})\cdot(1+\sum_{t=1}^{i}w^{t-1})}{q^{j}\cdot\sum_{t=i}^{j+1}w^{t-1}+(q^{j}-q^{i-2})\cdot(1+\sum_{t=1}^{i}w^{t-1})} \\
    &=\underbrace{\left(\frac{w^{i-1}}{\sum_{t=i}^{j}w^{t-1}}+w\right)}_{\le 3 \textrm{ since } w=2}\cdot\frac{q^{j-1}\cdot\frac{w^{i-1}(w^{j-i+1}-1)}{w-1}+(q^{j-1}-q^{i-2})\cdot(1+\frac{w^{i-1}-1}{w-1})}{q^{j}\cdot \frac{w^{i-1}(w^{j-i+2}-1)}{w-1}+(q^{j}-q^{i-2})\cdot(1+\frac{w^{i-1}-1}{w-1})} \\
    &\le 3\cdot\frac{q^{j-1}\cdot(2^{j-i+1}-1)+(q^{j-1}-q^{i-2})}{q^{j}\cdot(2^{j-i+2}-1)+(q^{j}-q^{i-2})} \quad \textrm{(Plugging in $w=2$)} \\
    &=\frac{3}{q}\cdot\frac{(2^{j-i+1}-1)+(1-q^{i-j-1})}{(2^{j-i+2}-1)+(1-q^{i-j-2})} \\
    &=\frac{3}{2q}\cdot\frac{1-(2q)^{i-j-1}}{1-(2q)^{i-j-2}} \\
    &\le \frac{3}{2q} \quad \textrm{(Since $\frac{1-(2q)^{-x}}{1-(2q)^{-x-1}}$ for $x\ge0$ is increasing and its limit is $1$)} \\
    &= \frac{3}{2+\sqrt{2}} .
\end{align*}
We notice that for $i=1$, by definition, $e_{1,j}=\frac{w^{j}-1}{q^{j-1}\cdot w^{j}}$, which is also decreasing in $j$.
Now let $B^{\textrm{greedy}}_k$ denote the minimum sufficient budget for \textsc{Greedy} to buy all the items in group $0$ to $k$. Clearly, by Myerson's payment rule, the payment to each unit of utility is $\gamma_k$. Thus, $B^{\textrm{greedy}}_k=\gamma_k(1+\sum_{i=1}^k w^{i-1})=q^{k-1}(1+(w^k-1)/(w-1))=q^{k-1}w^k$ (by $w=2$). Obviously, $B^{\textrm{greedy}}_k>B_k$, since the mechanism is overpaying the zeroth to $(k-1)$-th groups. On the other hand, for $k\ge 1$,
\begin{align*}
    B_{k+1}&=\sum_{j=1}^{k+1} w^{j-1}q^{j-1} \\
    &=\frac{(wq)^{k+1}-1}{wq-1}\\
    &=((wq)^{k+1}-1)/(1+\sqrt{2}) && \textrm{(By $q=1+1/\sqrt{2},w=2$)} \\
    &>q^kw^{k+1}/(1+\sqrt{2}) && \textrm{($qx-1>x$ for any $x>1/(q-1)$)}\\
    &>q^{k-1}w^k \\
    &=B^{\textrm{greedy}}_k.
\end{align*}
For a budget $B\in[B^{\textrm{greedy}}_k,B_{k+1}]$, after spending budget $B^{\textrm{greedy}}_{k}$, \textsc{Greedy} starts to buy items from $(k+1)$-th group and has earning rate $e_{k+1,k+1}$, while non-IC optimal solution is also buying from $(k+1)$-th group and has earning rate $1/q^{k}$, which is strictly better than $e_{k+1,k+1}$. For a budget $B\in[B_{k},B^{\textrm{greedy}}_k]$, after spending budget $B_{k}$, \textsc{Greedy} is still buying items from $k$-th group and has earning rate $e_{k,k}$, while non-IC optimum already starts to buy items from $(k+1)$-th group and has earning rate $1/q^{k}$. In the following, we show that when $k\ge 2$, $1/q^{k}$ is worse than $e_{k,k}$.
\begin{align}
    &e_{k,k}/(1/q^k)=q^ke_{k,k} \nonumber\\
    &=\frac{q^kw^{k-1}}{q^{k-1}w^{k-1}+(q^{k-1}-q^{k-2})\cdot(1+\sum_{1\le t\le k-1}w^{t-1})} && \textrm{(By definition of $e_{k,k}$)} \nonumber\\
    &=\frac{q^kw^{k-1}}{q^{k-1}w^{k-1}+(q^{k-1}-q^{k-2})\cdot(1+(w^{k-1}-1)/(w-1))} \nonumber\\
    &=\frac{q^kw^{k-1}}{q^{k-1}w^{k-1}+(q^{k-1}-q^{k-2})\cdot w^{k-1}} && \textrm{(By $w=2$)} \nonumber\\
    &=\frac{q}{1+(1-1/q)}>1.
\label{eq:compare_efficiency}
\end{align}
Next, consider the competitive ratio as budget $B$ increases continuously from $B_2$. When $B\in[B_{k},B^{\textrm{greedy}}_k]$, \textsc{Greedy} is purchasing more efficiently than the non-IC optimum, and hence the competitive ratio increases, and when $B\in[B^{\textrm{greedy}}_k,B_{k+1}]$, \textsc{Greedy} is less efficient. Suppose the non-IC optimum and the utility achieved by \textsc{Greedy} for budget $b$ are $OPT_{b}$ and $ALG_{b}$, then the competitive ratio at $B\in[B^{\textrm{greedy}}_k,B_{k+1}]$ is 
\[
    \frac{ALG_{B^{\textrm{greedy}}_k}+(B-B^{\textrm{greedy}}_k)\cdot e_{k+1,k+1}}{OPT_{B^{\textrm{greedy}}_k}+(B-B^{\textrm{greedy}}_k)/q^{k}},
\]
and notice that the ratio is either monotone increasing or monotone decreasing in $B$. If the ratio is increasing for $B\in[B^{\textrm{greedy}}_k,B_{k+1}]$, then it must be upper bounded by $e_{k+1,k+1}/(1/q^{k})$, which is $1/(1+(1-1/q))=1/\sqrt{2}<0.854$ by Eq.~\eqref{eq:compare_efficiency}. 

If otherwise, then the competitive ratio peaks at budget $B^{\textrm{greedy}}_k$. It remains to upper bound the competitive ratio at budget $B^{\textrm{greedy}}_k$. This competitive ratio is
\begin{align*}
    &\frac{1+\sum_{i=1}^k w^{i-1}}{1+\sum_{i=1}^k w^{i-1}+(B^{\textrm{greedy}}_k-B_k)/q^{k}} \\
    &=\frac{w^{k}}{w^{k}+(B^{\textrm{greedy}}_k-B_k)/q^{k}} \\
    &=\frac{w^{k}}{w^{k}+(w^kq^{k-1}-\frac{w^kq^k-1}{wq-1})/q^{k}} \\
    &\le\frac{w^{k}}{w^{k}+(w^kq^{k-1}-\frac{w^kq^k}{wq-1})/q^{k}} \\
    &=\frac{1}{1+(\frac{1}{q}-\frac{1}{wq-1})}\\
    &=\frac{2+\sqrt{2}}{4}<0.854. \\
\end{align*}

\end{proof}

\bibliographystyle{alpha}
\bibliography{cite}

\newcommand{\etalchar}[1]{$^{#1}$}
\begin{thebibliography}{OMV{\etalchar{+}}20}

\bibitem[AGN18]{anari2018budget}
Nima Anari, Gagan Goel, and Afshin Nikzad.
\newblock Budget feasible procurement auctions.
\newblock {\em Operations Research}, 66(3):637--652, 2018.

\bibitem[AKS19]{AmanatidisKS19}
Georgios Amanatidis, Pieter Kleer, and Guido Sch{\"{a}}fer.
\newblock Budget-feasible mechanism design for non-monotone submodular
  objectives: Offline and online.
\newblock In {\em Proceedings of the 2019 {ACM} Conference on Economics and
  Computation, {EC} 2019, Phoenix, AZ, USA, June 24-28, 2019}, pages 901--919,
  2019.

\bibitem[BGG{\etalchar{+}}22]{BGGST21}
Eric Balkanski, Pranav Garimidi, Vasilis Gkatzelis, Daniel Schoepflin, and
  Xizhi Tan.
\newblock Deterministic budget-feasible clock auctions.
\newblock In {\em Proceedings of the 2022 Annual ACM-SIAM Symposium on Discrete
  Algorithms (SODA)}, pages 2940--2963. SIAM, 2022.

\bibitem[BH16]{BalkanskiH16}
Eric Balkanski and Jason~D. Hartline.
\newblock Bayesian budget feasibility with posted pricing.
\newblock In {\em Proceedings of the 25th International Conference on World
  Wide Web, {WWW} 2016, Montreal, Canada, April 11 - 15, 2016}, pages 189--203,
  2016.

\bibitem[BKS12]{BadanidiyuruKS12}
Ashwinkumar Badanidiyuru, Robert Kleinberg, and Yaron Singer.
\newblock Learning on a budget: posted price mechanisms for online procurement.
\newblock In {\em Proceedings of the 13th {ACM} Conference on Electronic
  Commerce, {EC} 2012, Valencia, Spain, June 4-8, 2012}, pages 128--145, 2012.

\bibitem[CC14]{ChanC14}
Hau Chan and Jing Chen.
\newblock Truthful multi-unit procurements with budgets.
\newblock In {\em Web and Internet Economics - 10th International Conference,
  {WINE} 2014, Beijing, China, December 14-17, 2014. Proceedings}, pages
  89--105, 2014.

\bibitem[CC16]{ChanC16}
Hau Chan and Jing Chen.
\newblock Budget feasible mechanisms for dealers.
\newblock In {\em Proceedings of the 2016 International Conference on
  Autonomous Agents {\&} Multiagent Systems, Singapore, May 9-13, 2016}, pages
  113--122, 2016.

\bibitem[CGH{\etalchar{+}}96]{corless1996lambertw}
Robert~M Corless, Gaston~H Gonnet, David~EG Hare, David~J Jeffrey, and Donald~E
  Knuth.
\newblock On the lambertw function.
\newblock {\em Advances in Computational mathematics}, 5(1):329--359, 1996.

\bibitem[CGL11]{chen2011approximability}
Ning Chen, Nick Gravin, and Pinyan Lu.
\newblock On the approximability of budget feasible mechanisms.
\newblock In {\em Proceedings of the twenty-second annual ACM-SIAM symposium on
  Discrete Algorithms}, pages 685--699. Society for Industrial and Applied
  Mathematics, 2011.

\bibitem[DFI18]{DFI18}
Djellel~Eddine Difallah, Elena Filatova, and Panos Ipeirotis.
\newblock Demographics and dynamics of mechanical turk workers.
\newblock In {\em Proceedings of the Eleventh {ACM} International Conference on
  Web Search and Data Mining, {WSDM} 2018, Marina Del Rey, CA, USA, February
  5-9, 2018}, pages 135--143, 2018.

\bibitem[DPS11]{DobzinskiPS11}
Shahar Dobzinski, Christos~H. Papadimitriou, and Yaron Singer.
\newblock Mechanisms for complement-free procurement.
\newblock In {\em Proceedings 12th {ACM} Conference on Electronic Commerce
  (EC-2011), San Jose, CA, USA, June 5-9, 2011}, pages 273--282, 2011.

\bibitem[EG14]{EG14}
Ludwig Ensthaler and Thomas Giebe.
\newblock A dynamic auction for multi-object procurement under a hard budget
  constraint.
\newblock {\em Research Policy}, 43(1):179 -- 189, 2014.

\bibitem[GJLZ19]{gravin2019optimal}
Nick Gravin, Yaonan Jin, Pinyan Lu, and Chenhao Zhang.
\newblock Optimal budget-feasible mechanisms for additive valuations.
\newblock In {\em Proceedings of the 2019 ACM Conference on Economics and
  Computation}, pages 887--900. ACM, 2019.

\bibitem[GNS14]{GoelNS14}
Gagan Goel, Afshin Nikzad, and Adish Singla.
\newblock Mechanism design for crowdsourcing markets with heterogeneous tasks.
\newblock In {\em Proceedings of the Seconf {AAAI} Conference on Human
  Computation and Crowdsourcing, {HCOMP} 2014, November 2-4, 2014, Pittsburgh,
  Pennsylvania, {USA}}, 2014.

\bibitem[HAM{\etalchar{+}}18]{HaraAMSCB18}
Kotaro Hara, Abigail Adams, Kristy Milland, Saiph Savage, Chris
  Callison{-}Burch, and Jeffrey~P. Bigham.
\newblock A data-driven analysis of workers' earnings on amazon mechanical
  turk.
\newblock In Regan~L. Mandryk, Mark Hancock, Mark Perry, and Anna~L. Cox,
  editors, {\em Proceedings of the 2018 {CHI} Conference on Human Factors in
  Computing Systems, {CHI} 2018, Montreal, QC, Canada, April 21-26, 2018}, page
  449. {ACM}, 2018.

\bibitem[HAM{\etalchar{+}}19]{HaraAMSHBC19}
Kotaro Hara, Abigail Adams, Kristy Milland, Saiph Savage, Benjamin~V. Hanrahan,
  Jeffrey~P. Bigham, and Chris Callison{-}Burch.
\newblock Worker demographics and earnings on amazon mechanical turk: An
  exploratory analysis.
\newblock In Regan~L. Mandryk, Stephen~A. Brewster, Mark Hancock, Geraldine
  Fitzpatrick, Anna~L. Cox, Vassilis Kostakos, and Mark Perry, editors, {\em
  Extended Abstracts of the 2019 {CHI} Conference on Human Factors in Computing
  Systems, {CHI} 2019, Glasgow, Scotland, UK, May 04-09, 2019}. {ACM}, 2019.

\bibitem[HIM14]{HorelIM14}
Thibaut Horel, Stratis Ioannidis, and S.~Muthukrishnan.
\newblock Budget feasible mechanisms for experimental design.
\newblock In {\em {LATIN} 2014: Theoretical Informatics - 11th Latin American
  Symposium, Montevideo, Uruguay, March 31 - April 4, 2014. Proceedings}, pages
  719--730, 2014.

\bibitem[Hoe94]{hoeffding1994probability}
Wassily Hoeffding.
\newblock Probability inequalities for sums of bounded random variables.
\newblock In {\em The Collected Works of Wassily Hoeffding}, pages 409--426.
  Springer, 1994.

\bibitem[HPC18]{hauser2018}
David Hauser, Gabriele Paolacci, and Jesse~J Chandler.
\newblock Common concerns with mturk as a participant pool: Evidence and
  solutions, Sep 2018.

\bibitem[KT18]{KhalilabadiT18}
Pooya~Jalaly Khalilabadi and {\'{E}}va Tardos.
\newblock Simple and efficient budget feasible mechanisms for monotone
  submodular valuations.
\newblock In {\em Web and Internet Economics - 14th International Conference,
  {WINE} 2018, Oxford, UK, December 15-17, 2018, Proceedings}, pages 246--263,
  2018.

\bibitem[LMSZ17]{LeonardiMSZ17}
Stefano Leonardi, Gianpiero Monaco, Piotr Sankowski, and Qiang Zhang.
\newblock Budget feasible mechanisms on matroids.
\newblock In {\em Integer Programming and Combinatorial Optimization - 19th
  International Conference, {IPCO} 2017, Waterloo, ON, Canada, June 26-28,
  2017, Proceedings}, pages 368--379, 2017.

\bibitem[LZY20]{LiZY20}
Juan Li, Yanmin Zhu, and Jiadi Yu.
\newblock Redundancy-aware and budget-feasible incentive mechanism in crowd
  sensing.
\newblock {\em Comput. J.}, 63(1):66--79, 2020.

\bibitem[Mic20]{microworkersdata}
Microworkers.
\newblock Budget data of microworkers.
\newblock Private communication, 2020.

\bibitem[Mye81]{myerson1981optimal}
Roger~B Myerson.
\newblock Optimal auction design.
\newblock {\em Mathematics of operations research}, 6(1):58--73, 1981.

\bibitem[NSKK16]{NushiS0K16}
Besmira Nushi, Adish Singla, Andreas Krause, and Donald Kossmann.
\newblock Learning and feature selection under budget constraints in
  crowdsourcing.
\newblock In {\em Proceedings of the Fourth {AAAI} Conference on Human
  Computation and Crowdsourcing, {HCOMP} 2016, 30 October - 3 November, 2016,
  Austin, Texas, {USA}}, pages 159--168, 2016.

\bibitem[OMV{\etalchar{+}}20]{OMVI20}
Jonas Oppenlaender, Kristy Milland, Aku Visuri, Panos Ipeirotis, and Simo
  Hosio.
\newblock Creativity on paid crowdsourcing platforms.
\newblock In {\em {CHI} '20: {CHI} Conference on Human Factors in Computing
  Systems, Honolulu, HI, USA, April 25-30, 2020}, pages 1--14, 2020.

\bibitem[PBL{\etalchar{+}}19]{PBLDFS19}
Lisa Posch, Arnim Bleier, Clemens Lechner, Daniel Danner, Fabian Fl{\"{o}}ck,
  and Markus Strohmaier.
\newblock Measuring motivations of crowdworkers: The multidimensional
  crowdworker motivation scale.
\newblock {\em {ACM} Trans. Soc. Comput.}, 2(2):8:1--8:34, 2019.

\bibitem[RZ22]{RZ20}
Aviad Rubinstein and Junyao Zhao.
\newblock {Budget-Smoothed Analysis for Submodular Maximization}.
\newblock In Mark Braverman, editor, {\em 13th Innovations in Theoretical
  Computer Science Conference (ITCS 2022)}, volume 215 of {\em Leibniz
  International Proceedings in Informatics (LIPIcs)}, pages 113:1--113:23,
  Dagstuhl, Germany, 2022. Schloss Dagstuhl -- Leibniz-Zentrum f{\"u}r
  Informatik.

\bibitem[S{\etalchar{+}}58]{sion1958general}
Maurice Sion et~al.
\newblock On general minimax theorems.
\newblock {\em Pacific Journal of mathematics}, 8(1):171--176, 1958.

\bibitem[Sin10]{singer2010budget}
Yaron Singer.
\newblock Budget feasible mechanisms.
\newblock In {\em 2010 IEEE 51st Annual Symposium on Foundations of Computer
  Science}, pages 765--774. IEEE, 2010.

\bibitem[SM13]{SingerM13}
Yaron Singer and Manas Mittal.
\newblock Pricing mechanisms for crowdsourcing markets.
\newblock In {\em 22nd International World Wide Web Conference, {WWW} '13, Rio
  de Janeiro, Brazil, May 13-17, 2013}, pages 1157--1166, 2013.

\bibitem[ZLM16]{ZhaoLM16}
Dong Zhao, Xiang{-}Yang Li, and Huadong Ma.
\newblock Budget-feasible online incentive mechanisms for crowdsourcing tasks
  truthfully.
\newblock {\em {IEEE/ACM} Trans. Netw.}, 24(2):647--661, 2016.

\bibitem[ZWG{\etalchar{+}}17]{ZhengWGZTC17}
Zhenzhe Zheng, Fan Wu, Xiaofeng Gao, Hongzi Zhu, Shaojie Tang, and Guihai Chen.
\newblock A budget feasible incentive mechanism for weighted coverage
  maximization in mobile crowdsensing.
\newblock {\em {IEEE} Trans. Mob. Comput.}, 16(9):2392--2407, 2017.

\end{thebibliography}

\appendix
\section{Motivation of budget-smoothed analysis}\label{sec:smoothed_motivation}
In this section, we give an in-depth discussion on the motivation of budget-smoothed analysis in the context of budget-feasible mechanism design for ``beyond-worst-case instances''.

When modeling ``beyond-worst-case instances'', there are two dimensions to consider -- the market and the budget. Ideally, we want the family of beyond-worst-case instances to capture all the real-world applications (e.g., crowdsourcing platforms). Understanding the demographics of workers on crowdsourcing platforms is a complicated ongoing endeavor~\cite{DFI18,HaraAMSCB18,hauser2018,HaraAMSHBC19,PBLDFS19,OMVI20}, and the takeaway is that modeling beyond-worst-case market is tricky and application-dependent, and moreover, it is often hard to verify the model assumptions in practice. In contrast, calculating the empirical distribution of the budgets in practice is a much more tractable task. For example, it was easy for the crowdsourcing platform Microworkers to provide us with the budget distribution of the employers on their platform%
\footnote{We are grateful to Microworkers for sharing this dataset with us.}. Therefore, instead of making controversial assumptions on the real-world markets, in the budget-smoothed analysis, we model ``beyond-worst-case instances''  by modeling the beyond-worst-case behavior of the much simpler object---budget constraint\footnote{We note that beyond-worst-case model of budget constraint is orthogonal to any assumptions about the real-world markets, and in principle could be combined with any of them to obtain even stronger results.}. As we now explain, this has motivations from both theoretical and practical perspectives.

In practice, many applications of budget-feasible mechanisms have multiple buyers with similar objectives operating on the same, possibly the worst-case market. However, their
budgets can easily vary by an order of magnitude or more because of different sizes of business
or different amounts of funds. A concrete example is the budget distribution of the top 10 employers on Microworkers (Table~\ref{table:data}). Thus even if the market is worst-case, the ``average'' employer uses an ``average'' budget that is independent of the market. Therefore, from a practical
perspective, it is interesting to understand whether a widespread distribution\footnote{Tiny perturbations certainly cannot escape the worst-case hardness.} of budgets allows a mechanism to achieve better-than-worst-case utility for the ``average'' employers. In theory, the existing worst-case instances of~\cite{anari2018budget} break under a large budget perturbation: If we perturb (i.e., multiply) the budget by a significant factor like $2.5$, then simply setting a single uniform price (i.e.,~\cite{EG14}'s open clock auction) will achieve $100\%$ of the optimal utility, which is much better than the best possible $1-1/e$ competitive ratio for the worst-case budget. Therefore, from a theoretical perspective, it is interesting to study the effect perturbing the budget has on the
hardness results. With the above motivations from theory and practice, as part of our work, we investigate the semi-adversarial model---budget-smoothed analysis (formally defined in Section~\ref{sec:budget-smoothed-model}), and we hope to answer the following questions:
\begin{description}
    \item[Question 1] What is the optimal mechanism for a given budget distribution? Do our proposed mechanism and~\cite{anari2018budget}'s worst-case optimal mechanism continue to be optimal?
    \item[Question 2] Does a widespread budget distribution allow the optimal mechanism to achieve better-than-worst-case competitive ratio in expectation?
    \item[Question 3] Is there a strong hardness result that is robust to any budget perturbations?
\end{description}

\section{Useful lemmata}
The following standard lemmata and simple facts are used in our analysis.
\begin{lemma}[Hoeffding Bound~\cite{hoeffding1994probability}]\label{lem:hoeffding}
Let $X=\sum_{i=1}^n X_i$ where $X_i$'s are independent random variables in the interval $[a_i,b_i]$, then $\P[|X-\E[X]|\ge t]\le 2\exp\left(-\frac{2t^2}{\sum_{i=1}^n(b_i-a_i)^2}\right)$.
\end{lemma}

\begin{fact}\label{lem:sos_inequality}
Given a sequence $u_1,\dots u_n\in [0,b]$ s.t. $\sum_{i\in[n]} u_i = U$, we have that $\sum_{i\in[n]}u_i^2\le Ub$.
\end{fact}
\begin{proof}
We prove this using an exchange argument. (1) Suppose there exist distinct $i, j$ such that $u_i+u_j\le b$, then removing $u_i,u_j$ from the sequence and inserting $u_i+u_j$ to the sequence can only increase the sum of squares of the numbers in the sequence, because $(u_i+u_j)^2\ge u_i^2+u_j^2$. (2) Suppose there exist distinct $i, j$ such that $u_i,u_j<b$ and $u_i+u_j>b$, then removing $u_i,u_j$ from the sequence and inserting $b$ and $u_i+u_j-b$ to the sequence can only increase the sum of squares of the numbers in the sequence. Indeed, it is elementary to verify that the quadratic function $g(x)=x^2+(u_i+u_j-x)^2=2x^2-2(u_i+u_j)x+(u_i+u_j)^2$ over $[u_i+u_j-b, b]$ has maximum value at $u_i+u_j-b$ and $b$.

Note that whenever there exist distinct $i, j$ such that $u_i,u_j<b$, we can apply one of the above two steps (depending on whether $u_i+u_j\le b$) to strictly decrease the number of $u_{\ell}$'s in the sequence that are strictly less than $b$, and the sum of squares of the numbers in the sequence can only increase. Since there are only $n$ numbers, this process will terminate when there is at most one number strictly less than $b$ in the sequence. Therefore, the sum of squares of the numbers in the final sequence, which is an upper bound for the sum of squares of the numbers in the original sequence, is
\[\floor{\frac{U}{b}}b^2+(U-\floor{\frac{U}{b}}b)^2=\floor{\frac{U}{b}}b^2+(\frac{U}{b}-\floor{\frac{U}{b}})^2b^2\le \floor{\frac{U}{b}}b^2+(\frac{U}{b}-\floor{\frac{U}{b}})b^2 = Ub.\]
\end{proof}

\begin{fact}\label{lem:product_sum_inequality}
For any $a,b\ge 2$, $ab\ge a+b$.
\end{fact}
\begin{proof}
Assume without loss of generality $a\ge b$. We have that $ab\ge 2a\ge a+b$.
\end{proof}

\begin{fact}\label{lem:sum_of_double_exponential}
For any $\delta>0$ and integer $k\ge 1$, $\sum_{i\in[k]} \exp(-2(1+\delta)^{i-1})\le 1+\frac{1}{\delta}$.
\end{fact}
\begin{proof}
Because $(1+\delta)^{i-1}\ge 1+(i-1)\ln(1+\delta)$, we have that $\sum_{i\in[k]} \exp(-2(1+\delta)^{i-1})\le\sum_{i\in[k]} \exp(-2(i-1)\ln(1+\delta))$, and moreover,
\begin{equation*}
\sum_{i\in[k]} \exp(-2(i-1)\ln(1+\delta))=\sum_{i\in[k]}(1+\delta)^{-2(i-1)}=\frac{1-(1+\delta)^{-2k}}{1-(1+\delta)^{-2}}\le\frac{1}{1-(1+\delta)^{-1}}=1+\frac{1}{\delta}.
\end{equation*}
\end{proof}

\begin{lemma}[Diminishing marginal returns for fractional knapsack~\cite{anari2018budget}]\label{lem:knapsack_concave}
Let $U(V,B)$ be the optimal utility of a fractional knapsack of items $V$ with budget $B$. Then when $V$ is fixed, $U$ is a concave function in $B$, and in particular, $U(V,(1-\theta)B)\ge (1-\theta)U(V,B)$ for $0\le\theta\le1$.
\end{lemma}
\begin{proof}
Given $B=(1-\theta)B_1+\theta B_2$ for $0\le\theta\le1$, we let $x^*_1,x^*_2$ be the optimal solution for budget $B_1$ and $B_2$ respectively. Now consider the solution $x=(1-\theta)x^*_1+\theta x^*_2$. Because the utility and the payment are both additive, the utility of $x$ is $(1-\theta)U(V,B_1)+\theta U(V,B_2)$, and the payment of $x$ is $(1-\theta)B_1+\theta B_2=B$. Hence $x$ is feasible for budget $B$, and it follows that $U(V,B)\ge (1-\theta)U(V,B_1)+\theta U(V,B_2)$.
\end{proof}
\begin{lemma}[Minimax theorem~\cite{sion1958general}]\label{lem:min-max}
Let $X$ be a compact convex subset of a linear topological space and $Y$ a compact convex subset of a linear topological space. If $f$ is a real-valued function on $X\times Y$ with $f(x,\cdot )$ upper semi-continuous and quasi-concave on $Y$, $\forall x\in X$, and $f(\cdot ,y)$ lower semi-continuous and quasi-convex on $X$, $\forall y\in Y$, then, $\inf _{x\in X}\sup _{y\in Y}f(x,y)=\sup _{y\in Y}\inf _{x\in X}f(x,y)$. This implies that every equilibrium for the min-max program $\inf_{x\in X}\sup _{y\in Y}f(x,y)$ achieves optimal value. (An equilibrium is a solution pair $(x^*,y^*)$ for which $\forall y\in Y,\,f(x^*,y^*)\ge f(x^*,y)$ and $\forall x\in X,\,f(x^*,y^*)\le f(x,y^*)$.)
\end{lemma}

\section{Analysis of \textsc{Random-Sampling-Greedy}}\label{section:proof_random_sample_greedy}
First, we define a subclass of greedy allocation rules, for which we will prove concentration bounds in Lemma~\ref{lem:concentration_tau_large}, and then, we will prove Theorem~\ref{thm:random_sample_greedy}.
\begin{definition}
Given an allocation rule $f$ of \textsc{Greedy} that is characterized by $(t, p_1, p_2)$ where $t\in[0,1)$ and $0^-\le p_1\le p_2$, for $\tau>0$, we say that $f$ is $\tau$-large if it satisfies the following properties:
\begin{enumerate}
    \item[(i)] $U_{p_2}([n]\setminus T)\ge \tau U(T)$;
    \item[(ii)] If $p_1\ge 0$, then $U_{p_1}([n]\setminus T)\ge \tau U(T)$.
\end{enumerate}
\end{definition}
The following observation will be useful in the proof of Theorem~\ref{thm:random_sample_greedy}.
\begin{observation}\label{obs:truncated_rule_is_tau_large}
Given an allocation rule $f$ of \textsc{Greedy}, we let $\hat{f}$ be the allocation rule we get after truncating $f$ using parameter $\eta$. If $U_{\hat{f}}([n]\setminus T)\ge \frac{U(T)}{\eta C}$, then $\hat{f}$ is $\frac{1}{\eta C}$-large.
\end{observation}
\begin{proof}
Suppose $\hat{f}$ is characterized by $(t,\hat{p_1},p_2)$, then because $U_{\hat{f}}([n]\setminus T)\ge \frac{U(T)}{\eta C}$, and $U_{p_2}([n]\setminus T)\ge U_{\hat{f}}([n]\setminus T)$ by Eq.~\eqref{eq:f_to_ps} and monotonicity of $U_p(S)$ with respect to $p$, it follows that $U_{p_2}([n]\setminus T)\ge \frac{U(T)}{\eta C}$. Moreover, the truncation step also makes sure $U_{\hat{p_1}}([n]\setminus T)\ge\frac{U(T)}{\eta C}$ if $\hat{p_1}\ge0$.
\end{proof}

\begin{lemma}\label{lem:concentration_tau_large}
Assuming $\delta^2\tau C\ge 8$. With probability at least $1-4\exp\left(-\frac{\delta^2 \tau C}{4}\right)\cdot \left(1+\frac{1}{\delta}\right)$ the following holds {\em simultaneously} for every $\tau$-large greedy allocation rule $f$:
\begin{align*}
    \frac{(1-\delta)U_f([n]\setminus T)}{2(1+\delta)} \le U_f(X)\le \frac{(1+\delta)^2U_f([n]\setminus T)}{2},\\
    \frac{(1-\delta)B_f([n]\setminus T)}{2(1+\delta)} \le B_f(X)\le \frac{(1+\delta)^2B_f([n]\setminus T)}{2},
\end{align*}
where $X$ is the uniformly random subset of $[n]\setminus T$ given in Mechanism~\ref{mech:RS_greedy}.
\end{lemma}
\begin{proof}
Suppose $f$ is characterized by $(t,p_1,p_2)$.
By Observation~\ref{obs:utility_and_payment}, both $U_f(X)$ and $B_f(X)$ are non-negative linear combination of $U_{p_1}(X)$ and $U_{p_2}(X)$, and hence, it suffices to prove concentration inequalities for $U_{p_1}(X)$ and $U_{p_2}(X)$. Moreover, because $f$ is $\tau$-large, we only need to prove concentration inequality for $U_{p}(X)$, for all $p$ such that $U_p([n])\ge \tau U(T)$ and for all $p<0$
The case of $p<0$ is trivial because $U_p([n])=0$ (as all the seller's costs are non-negative), and thus, it remains to prove the concentration bound for $U_{p}(X)$, for all $p$ such that $U_p([n]\setminus T)\ge \tau U(T)$.

\subsubsection*{Discretization and the high-level proof plan}
We discretize the space of all the $p$ such that $U_p([n]\setminus T)\ge \tau U(T)$. Specifically, for $i\ge 1$, we let $p^{(i)}_{\min}$ be the smallest $\frac{c_j}{u_j}$ among all $j\in[n]\setminus T$ such that $U_{\frac{c_j}{u_j}}([n]\setminus T)\in [(1+\delta)^{i-1} \tau U(T), (1+\delta)^{i} \tau U(T))$, and similarly, we let $p^{(i)}_{\max}$ be the largest such $\frac{c_j}{u_j}$ among all $j\in[n]\setminus T$. We let $k$ be largest integer such that $U_{\infty}([n]\setminus T)\ge(1+\delta)^{k-1} \tau U(T)$ (we will only consider $k\ge 1$ because if $k<1$, then $U_{\infty}([n]\setminus T)<\tau U(T)$, which implies that there is no $\tau$-large allocation rule by definition of $\tau$-largeness, and this lemma is trivially true).

The high-level proof plan is as follows: We first prove a concentration bound for all $U_{p^{(i)}_{\min}}(X)$ and $U_{p^{(i)}_{\max}}(X)$, and then we argue by monotonicity that $U_p(X)$ for all $p\in(p^{(i)}_{\min}, p^{(i)}_{\max})$ is also concentrated because of our discretization. Notice that we do not need to prove concentration bounds for $p\in (p^{(i)}_{\max}, p^{(i+1)}_{\min})$ (or $p\in (p^{(k)}_{\max},\infty)$ in the case that $i=k$),  
because there are no sellers with such cost-per-utility. Similarly, we do not need the concentration if $p^{(i)}_{\min}$ and $p^{(i)}_{\max}$ do not exist for some $i$ because the $i$-th ``bucket'' is empty.

\subsubsection*{A concentration bound for all $U_{p^{(i)}_{\min}}(X)$ and $U_{p^{(i)}_{\max}}(X)$}
Now we prove the concentration bound for $U_{p^{(i)}_{\min}}(X)$ for all $i\in[k]$ (the proof of concentration bound for $U_{p^{(i)}_{\max}}(X)$ is analogous). Let $N_i$ be the set of sellers in $[n]\setminus T$ whose cost-per-utility is at most $p^{(i)}_{\min}$. Since $U(T)$ is the total utility of the top $C$ sellers, the utility $u_j$ of a seller $j\in[n]\setminus T$ is at most $U(T)/C$. Thus, we have that $u_j\le U(T)/C$ for all $j\in N_i$ and that $\sum_{j\in N_i} u_j=U_{p^{(i)}_{\min}}([n]\setminus T)$ (by definition of $N_i$), and then it follows by Lemma~\ref{lem:sos_inequality} that
\begin{equation}\label{eq:total_influence_bound}
\sum_{j\in N_i} u_j^2\le U_{p^{(i)}_{\min}}([n]\setminus T)\cdot U(T)/C.
\end{equation}
Since each seller in $[n]\setminus T$ is added to $X$ independently with probability $1/2$, we can apply the Hoeffding bound (Lemma~\ref{lem:hoeffding}) and get
\begin{align*}
    \P\left[\underbrace{\left|U_{p^{(i)}_{\min}}(X)-\frac{U_{p^{(i)}_{\min}}([n]\setminus T)}{2}\right|\ge \frac{\delta U_{p^{(i)}_{\min}}([n]\setminus T)}{2}}_{ =: \text{event $E^{(i)}_{\min}$}}\right] 
    \le& 2\exp\left(-\frac{\delta^2U_{p^{(i)}_{\min}}([n]\setminus T)^2}{2\sum_{j\in N_i}u_j^2}\right) \\
    \le& 2\exp\left(-\frac{\delta^2U_{p^{(i)}_{\min}}([n]\setminus T)}{2U(T)/C}\right) && \text{(By Eq.~\eqref{eq:total_influence_bound})} \\ 
    \le& 2\exp\left(-\frac{\delta^2(1+\delta)^{i-1} \tau C}{2}\right).
\end{align*}
Where the last inequality follows by $U_{p^{(i)}_{\min}}([n]\setminus T)\ge (1+\delta)^{i-1} \tau U(T)$.

Taking a union bound over all $i\in[k]$, we have that
\begin{align*}
    \P\left[\exists i\in[k],\,E^{(i)}_{\min}\right] 
    \le&2\sum_{i\in[k]}\exp\left(-\frac{\delta^2(1+\delta)^{i-1} \tau C}{2}\right)\\
    =&2\sum_{i\in[k]}\exp\left(-\frac{\delta^2 \tau C}{4}\cdot 2(1+\delta)^{i-1}\right)\\
    \le&2\sum_{i\in[k]}\exp\left(-\left(\frac{\delta^2 \tau C}{4}+ 2(1+\delta)^{i-1}\right)\right) &&\text{(Fact~\ref{lem:product_sum_inequality} and $\delta^2\tau C\ge 8$)}\\
    =&2\exp\left(-\frac{\delta^2 \tau C}{4}\right)\cdot\sum_{i\in[k]} \exp(-2(1+\delta)^{i-1}) \\
    \le&2\exp\left(-\frac{\delta^2 \tau C}{4}\right)\cdot \left(1+\frac{1}{\delta}\right) &&\text{(Fact~\ref{lem:sum_of_double_exponential})}.
\end{align*}
We also want to take a union bound for $U_{p^{(i)}_{\max}}(X)$, and we get
\begin{gather*}
    \P\left[\exists i\in[k],\,E^{(i)}_{\min} \textnormal{ or }  E^{(i)}_{\max}\right]\le 4\exp\left(-\frac{\delta^2 \tau C}{4}\right)\cdot \left(1+\frac{1}{\delta}\right).
\end{gather*}

\subsubsection*{Concentration bound for general $U_{p}(X)$ follows from monotonicity}
Next we show that if $U_{p^{(i)}}(X)\in\left[\frac{(1-\delta)U_{p^{(i)}}([n]\setminus T)}{2}, \frac{(1+\delta)U_{p^{(i)}}([n]\setminus T)}{2}\right]$ for $p^{(i)}=p^{(i)}_{\min}\textrm{ or }p^{(i)}_{\max}$ for all $i\in[k]$, then $U_{p}(X)\in\left[\frac{(1-\delta)U_{p}([n]\setminus T)}{2(1+\delta)}, \frac{(1+\delta)^2U_{p}([n]\setminus T)}{2}\right]$ for all $p$ such that $U_p(n\setminus T)\ge \tau U(T)$. This essentially follows from the monotonicity of $U_{p}(S)$ with respect to $p$, given any fixed $S\subseteq [n]$ (see Eq.~\eqref{eq:B_and_U_for_p}).

Specifically, for any $p\in (p^{(i)}_{\min}, p^{(i)}_{\max})$ with $i\in[k]$, by Eq.~\eqref{eq:B_and_U_for_p}, $U_{p^{(i)}_{\min}}(S)\le U_p(S)\le U_{p^{(i)}_{\max}}(S)$ for all $S\subseteq[n]$, and thus, $U_p(X)\le U_{p^{(i)}_{\max}}(X)\le \frac{(1+\delta)U_{p^{(i)}_{\max}}([n]\setminus T)}{2}$. Moreover, by definition of $p^{(i)}_{\min},p^{(i)}_{\max}$ (i.e., our discretization), we have that $U_{p^{(i)}_{\max}}([n]\setminus T)\le (1+\delta) U_{p^{(i)}_{\min}}([n]\setminus T)$, which implies that $\frac{(1+\delta)U_{p^{(i)}_{\max}}([n]\setminus T)}{2}\le\frac{(1+\delta)^2U_{p^{(i)}_{\min}}([n]\setminus T)}{2}\le \frac{(1+\delta)^2U_{p}([n]\setminus T)}{2}$. Therefore, it follows that $U_p(X)\le \frac{(1+\delta)^2U_{p}([n]\setminus T)}{2}$, and similarly, we can derive that $U_p(X)\ge \frac{(1-\delta)U_{p}([n]\setminus T)}{2(1+\delta)}$.
\end{proof}

Now we are ready to prove Theorem~\ref{thm:random_sample_greedy}.
\begin{proof}[Proof of Theorem~\ref{thm:random_sample_greedy}]
Note that \textsc{Random-Sampling-Greedy} is
truthful-in-expectation simply because each seller's allocation and payment rules are independent of his own private cost, and the allocation rule for each seller is monotone.
Moreover, it strictly satisfies the budget constraint, because it spends $\epsilon_1 B$ at line~\ref{algline:topseller} and at most $(1-\epsilon_1)B$ at line~\ref{algline:otherseller}, which in total is at most $B$. Throughout the rest of the proof, we let $\delta_2>0$ be arbitrarily small constant, and then we choose $C$ such that $\eta C$ is arbitrarily large and $\delta_2^2\tau C\ge 8$, and we choose $\delta_1,\epsilon_1,\delta_3>0$ such that $(1-\delta_3)(1+\delta_2)^2, \frac{(1+3\delta_2)(1-\delta_3)}{1+\delta_2}\le 1-\delta_1$ and $\frac{(1+3\delta_2)(1-\delta_1)}{(1-\delta_2)}, \frac{(1+\delta_2)^2(1-\delta_1)}{1-2\delta_2-\delta_2^2} \le 1-\epsilon_1$ (also note that all $\delta_1,\epsilon_1,\delta_3$ can be made arbitrarily small given sufficiently small $\delta_2$). Moreover, we {\bf always condition on the event $E$} that for all $\frac{1}{\eta C}$-large greedy allocation rule $f$,
\begin{align*}
    \frac{(1-\delta_2)U_f([n]\setminus T)}{2(1+\delta_2)} \le U_f(X)\le \frac{(1+\delta_2)^2U_f([n]\setminus T)}{2},\\
    \frac{(1-\delta_2)B_f([n]\setminus T)}{2(1+\delta_2)} \le B_f(X)\le \frac{(1+\delta_2)^2B_f([n]\setminus T)}{2},
\end{align*}
which happens with probability at least $1-4\exp\left(-\frac{\delta_2^2}{4\eta}\right)\cdot \left(1+\frac{1}{\delta_2}\right)$ by Lemma~\ref{lem:concentration_tau_large}. Since $\eta$ can be made arbitrarily small given any $\delta_2$, this is a high-probability event.

Let $f^*_{B}$ and $f^*_{(1-\delta_3)B}$ denote the optimal allocation rules that \textsc{Greedy} chooses for all the sellers $[n]$ with budget $B$ and $(1-\delta_3)B$ respectively. We note that $U_{f^*_{(1-\delta_3)B}}([n])\ge (1-\delta_3)U_{f^*_B}([n])$, because we can scale $f^*_B$ down by a multiplicative factor $1-\delta_3$ such that it achieves a total utility $(1-\delta_3)U_{f^*_B}([n])$ and makes a total payment $\le(1-\delta_3)B$, and it follows by optimality of $f^*_{(1-\delta_3)B}$ among uniform allocation rules with budget $(1-\delta_3)B$ (see Theorem~\ref{thm:greedy_instance_optimal}), $f^*_{(1-\delta_3)B}$ achieves no less utility than $(1-\delta_3)U_{f^*_B}([n])$.
Without loss of generality we assume that $U_{f^*_{(1-\delta_3)B}}([n])\ge U(T)$, because otherwise the first step of our mechanism already achieves nearly optimal utility. We apply the truncation step with parameter $\eta$ to $f^*_{(1-\delta_3)B}$ and get the truncated allocation rule $\widehat{f^*_{(1-\delta_3)B}}$. By Observation~\ref{obs:truncation_does_not_hurt}, we have that $U_{\widehat{f^*_{(1-\delta_3)B}}}([n])\ge U_{f^*_{(1-\delta_3)B}}([n])-\frac{U(T)}{\eta C}\ge(1-\frac{1}{\eta C})U_{f^*_{(1-\delta_3)B}}([n])$. Therefore, we can henceforth use $U_{\widehat{f^*_{(1-\delta_3)B}}}([n])$ as a benchmark, because it is arbitrarily close to $U_{f^*_B}([n])$ for sufficiently small $\delta_3$ and sufficiently large $\eta C$.

\subsubsection*{$f_X$, $f_Y$ achieve utility as good as $\widehat{f^*_{(1-\delta_3)B}}$ on $X$, $Y$ respectively}
We show that without loss of generality $\widehat{f^*_{(1-\delta_3)B}}$ is $\frac{1}{\eta C}$-large. Indeed, without loss of generality, we can assume that $U_{\widehat{f^*_{(1-\delta_3)B}}}([n]\setminus T)\ge \frac{U(T)}{\eta C}$, because otherwise $U_{\widehat{f^*_{(1-\delta_3)B}}}([n])=U_{\widehat{f^*_{(1-\delta_3)B}}}([n]\setminus T)+U_{\widehat{f^*_{(1-\delta_3)B}}}(T)\le (1+\frac{1}{\eta C})U(T)$, which is arbitrarily close to $U(T)$ for sufficiently large $\eta C$, and thus, the first step of Mechanism~\ref{mech:RS_greedy} already achieves nearly optimal utility. Given $U_{\widehat{f^*_{(1-\delta_3)B}}}([n]\setminus T)\ge \frac{U(T)}{\eta C}$, it follows by Observation~\ref{obs:truncated_rule_is_tau_large} that $\widehat{f^*_{(1-\delta_3)B}}$ is $\frac{1}{\eta C}$-large.

It follows from the event $E$ that
\begin{align*}
U_{\widehat{f^*_{(1-\delta_3)B}}}(X)\in\left[\frac{(1-\delta_2)U_{\widehat{f^*_{(1-\delta_3)B}}}([n]\setminus T)}{2(1+\delta_2)}, \frac{(1+\delta_2)^2U_{\widehat{f^*_{(1-\delta_3)B}}}([n]\setminus T)}{2}\right],\\ B_{\widehat{f^*_{(1-\delta_3)B}}}(X)\in\left[\frac{(1-\delta_2)B_{\widehat{f^*_{(1-\delta_3)B}}}([n]\setminus T)}{2(1+\delta_2)},\frac{(1+\delta_2)^2B_{\widehat{f^*_{(1-\delta_3)B}}}([n]\setminus T)}{2}\right].
\end{align*}

By definition of $\widehat{f^*_{(1-\delta_3)B}}$, $B_{\widehat{f^*_{(1-\delta_3)B}}}([n]\setminus T)\le B_{\widehat{f^*_{(1-\delta_3)B}}}([n])\le B_{f^*_{(1-\delta_3)B}}([n])\le (1-\delta_3)B$, and hence $B_{\widehat{f^*_{(1-\delta_3)B}}}(X)\le \frac{(1+\delta_2)^2(1-\delta_3)B}{2}$, which is at most $\frac{(1-\delta_1)B}{2}$ by our choice of $\delta_1,\delta_3$. Similarly, since $X\cup Y=[n]\setminus T$, we have that $B_{\widehat{f^*_{(1-\delta_3)B}}}(Y)=B_{\widehat{f^*_{(1-\delta_3)B}}}([n]\setminus T)-B_{\widehat{f^*_{(1-\delta_3)B}}}(X)\le \frac{(1+3\delta_2)B_{\widehat{f^*_{(1-\delta_3)B}}}([n]\setminus T)}{2(1+\delta_2)}\le \frac{(1+3\delta_2)(1-\delta_3)B}{2(1+\delta_2)}$, which is also at most $\frac{(1-\delta_1)B}{2}$ by our choice of $\delta_1$. Thus, by definition of $f_X$ and $f_Y$ in Mechanism~\ref{mech:RS_greedy} and the optimality of Mechanism~\ref{mech:greedy} (see Theorem~\ref{thm:greedy_instance_optimal}), $f_X$ and $f_Y$ must achieve no less utility than $\widehat{f^*_{(1-\delta_3)B}}$ on $X$ and $Y$ respectively, i.e.,
\begin{align*}
    &U_{f_X}(X)\ge U_{\widehat{f^*_{(1-\delta_3)B}}}(X)\ge \frac{(1-\delta_2)U_{\widehat{f^*_{(1-\delta_3)B}}}([n]\setminus T)}{2(1+\delta_2)},\\
    U_{f_Y}(Y)\ge U_{\widehat{f^*_{(1-\delta_3)B}}}(Y)&=U_{\widehat{f^*_{(1-\delta_3)B}}}([n]\setminus T)-U_{\widehat{f^*_{(1-\delta_3)B}}}(X)\ge \frac{(1-2\delta_2-\delta_2^2)U_{\widehat{f^*_{(1-\delta_3)B}}}([n]\setminus T)}{2}.
\end{align*}
Furthermore, it follows by Observation~\ref{obs:truncation_does_not_hurt} that
\begin{align}
    U_{\hat{f_X}}(X)\ge U_{f_X}(X)-\frac{U(T)}{\eta C}\ge \frac{(1-\delta_2)U_{\widehat{f^*_{(1-\delta_3)B}}}([n]\setminus T)}{2(1+\delta_2)}-\frac{U(T)}{\eta C},\nonumber\\
    U_{\hat{f_Y}}(Y)\ge U_{f_Y}(Y)-\frac{U(T)}{\eta C}\ge \frac{(1-2\delta_2-\delta_2^2)U_{\widehat{f^*_{(1-\delta_3)B}}}([n]\setminus T)}{2}-\frac{U(T)}{\eta C}.\label{eq:utility_of_hat_f_X_and_hat_f_Y}
\end{align}

Henceforth, we assume that $U_{\hat{f_X}}(X),U_{\hat{f_Y}}(Y)\ge \frac{U(T)}{\eta C}$. This is without loss of generality -- For example, if $U_{\hat{f_X}}(X)<\frac{U(T)}{\eta C}$, then by Eq.~\eqref{eq:utility_of_hat_f_X_and_hat_f_Y}, $U_{\widehat{f^*_{(1-\delta_3)B}}}([n]\setminus T)\le \frac{4(1+\delta_2)U(T)}{(1-\delta_2)\eta C}$. Moreover, since $U_{\widehat{f^*_{(1-\delta_3)B}}}(T)\le U(T)$, we have that $U_{\widehat{f^*_{(1-\delta_3)B}}}([n])=U_{\widehat{f^*_{(1-\delta_3)B}}}(T)+U_{\widehat{f^*_{(1-\delta_3)B}}}([n]\setminus T)\le \left(1+\frac{4(1+\delta_2)}{(1-\delta_2)\eta C}\right)U(T)$, which is arbitrarily close to $U(T)$ for sufficiently small $\delta_2$ and sufficiently large $\eta C$. In this case, the first step of Mechanism~\ref{mech:RS_greedy} already achieves nearly optimal utility. 

\subsubsection*{$\hat{f_X}$ achieves utility on $Y$ as good as that on $X$}
Now we show that $\hat{f_X}$ achieves utility on $Y$ as good as that on $X$ (and it can be proved analogously that $\hat{f_Y}$ achieves utility on $X$ as good as that on $Y$). First, because $U_{\hat{f_X}}([n]\setminus T)\ge U_{\hat{f_X}}(X)\ge \frac{U(T)}{\eta C}$, it follows by Observation~\ref{obs:truncated_rule_is_tau_large} that $\hat{f_X}$ is $\frac{1}{\eta C}$-large. Then, because of the event $E$, $U_{\hat{f_X}}(X)\le \frac{(1+\delta_2)^2U_{\hat{f_X}}([n]\setminus T)}{2}=\frac{(1+\delta_2)^2(U_{\hat{f_X}}(X)+U_{\hat{f_X}}(Y))}{2}$, which implies that $U_{\hat{f_X}}(Y)\ge \frac{1-2\delta_2-\delta_2^2}{(1+\delta_2)^2}U_{\hat{f_X}}(X)$ (analogously, because $U_{\hat{f_Y}}(X)\ge \frac{(1-\delta_2)U_{\hat{f_Y}}([n]\setminus T)}{2(1+\delta_2)}=\frac{(1-\delta_2)(U_{\hat{f_Y}}(X)+U_{\hat{f_Y}}(Y))}{2(1+\delta_2)}$, we have that $U_{\hat{f_Y}}(X)\ge \frac{1-\delta_2}{1+3\delta_2}U_{\hat{f_Y}}(Y)$). Then, it follows from Eq.~\eqref{eq:utility_of_hat_f_X_and_hat_f_Y} that
\begin{equation}\label{eq:utility_of_hat_f_X_on_Y}
    U_{\hat{f_X}}(Y)\ge \frac{1-2\delta_2-\delta_2^2}{(1+\delta_2)^2}\left(\frac{(1-\delta_2)U_{\widehat{f^*_{(1-\delta_3)B}}}([n]\setminus T)}{2(1+\delta_2)}-\frac{U(T)}{\eta C}\right),
\end{equation}
and analogously
\begin{equation}\label{eq:utility_of_hat_f_Y_on_X}
U_{\hat{f_Y}}(X)\ge \frac{1-\delta_2}{1+3\delta_2}\left(\frac{(1-2\delta_2-\delta_2^2)U_{\widehat{f^*_{(1-\delta_3)B}}}([n]\setminus T)}{2}-\frac{U(T)}{\eta C}\right).
\end{equation}

\subsubsection*{$\hat{f_X}$ and $\hat{f_Y}$ never reach the budget limit (assuming event $E$)}
Before we proceed, we note that we have already proved that Mechanism~\ref{mech:RS_greedy} is strictly budget-feasible at the beginning of the proof. \textbf{The purpose of this part} is only to show that the budget threshold at line~\ref{algline:otherseller} is not met (conditioned on $E$), in which case the mechanism achieves all the utility attainable by $\hat{f_X}$ and $\hat{f_Y}$, i.e., $U_{\hat{f_X}}(Y)+U_{\hat{f_Y}}(X)$.

Specifically, because of the event $E$, we get $B_{\hat{f_X}}(X)\ge \frac{(1-\delta_2)B_{\hat{f_X}}([n]\setminus T)}{2(1+\delta_2)}=\frac{(1-\delta_2)(B_{\hat{f_X}}(X)+B_{\hat{f_X}}(Y))}{2(1+\delta_2)}$, which implies that $B_{\hat{f_X}}(Y)\le \frac{1+3\delta_2}{1-\delta_2}B_{\hat{f_X}}(X)$ (analogously, because $B_{\hat{f_Y}}(X)\le \frac{(1+\delta_2)^2B_{\hat{f_Y}}([n]\setminus T)}{2}=\frac{(1+\delta_2)^2(B_{\hat{f_Y}}(X)+B_{\hat{f_Y}}(Y))}{2}$, we have that $B_{\hat{f_Y}}(X)\le \frac{(1+\delta_2)^2}{1-2\delta_2-\delta_2^2}B_{\hat{f_Y}}(Y)$). Since $B_{\hat{f_X}}(X)\le B_{f_X}(X)$ and $B_{\hat{f_Y}}(Y)\le B_{f_Y}(Y)$ are at most $\frac{(1-\delta_1) B}{2}$ by design of Mechanism~\ref{mech:RS_greedy}, it follows by our choice of $\delta_1,\epsilon_1$ that $B_{\hat{f_X}}(Y)\le \frac{(1+3\delta_2)(1-\delta_1)}{2(1-\delta_2)} B\le \frac{(1-\epsilon_1)B}{2}$ (analogously, $B_{\hat{f_Y}}(X)\le\frac{(1+\delta_2)^2(1-\delta_1)}{2(1-2\delta_2-\delta_2^2)}B\le \frac{(1-\epsilon_1)B}{2}$). Therefore, we will not exceed the budget limit at line~\ref{algline:otherseller}.

\subsubsection*{The total utility achieved by Mechanism~\ref{mech:RS_greedy} is nearly optimal}
Finally, we calculate the total utility achieved by Mechanism~\ref{mech:RS_greedy}. As we argued above, we get utility $U_{\hat{f_X}}(Y)+U_{\hat{f_Y}}(X)$ at line~\ref{algline:otherseller}. Moreover, we get utility $U(T)$ at line~\ref{algline:topseller}. Therefore, the total utility is
\begin{align*}
    &U_{\hat{f_X}}(Y)+U_{\hat{f_Y}}(X)+U(T)\\
    \ge& \frac{1-2\delta_2-\delta_2^2}{(1+\delta_2)^2}\left(\frac{(1-\delta_2)U_{\widehat{f^*_{(1-\delta_3)B}}}([n]\setminus T)}{2(1+\delta_2)}-\frac{U(T)}{\eta C}\right)\\
    &+\frac{1-\delta_2}{1+3\delta_2}\left(\frac{(1-2\delta_2-\delta_2^2)U_{\widehat{f^*_{(1-\delta_3)B}}}([n]\setminus T)}{2}-\frac{U(T)}{\eta C}\right)+U(T) &&\text{(By Eq.~\eqref{eq:utility_of_hat_f_X_on_Y} and Eq.~\eqref{eq:utility_of_hat_f_Y_on_X})},
\end{align*}
which is arbitrarily close to (i.e., up to a multiplicative factor that is arbitrarily close to 1) $$U_{\widehat{f^*_{(1-\delta_3)B}}}([n]\setminus T)+U(T)\ge U_{\widehat{f^*_{(1-\delta_3)B}}}([n]\setminus T)+U_{\widehat{f^*_{(1-\delta_3)B}}}(T)=U_{\widehat{f^*_{(1-\delta_3)B}}}([n]),$$ for sufficiently small $\delta_2$ and sufficiently large $\eta C$.
\end{proof}

\section{A robust hard instance for the AGN mechanism}~\label{section:agn_instance}
In a nutshell,~\cite{anari2018budget} designed a proxy to a uniform mechanism that requires knowledge of all the costs of the sellers. For a parameter $r$ to be determined later, this uniform mechanism is defined by a carefully chosen allocation function $f_r : \mathbb{R}_{\geq 0} \rightarrow [0,1]$, which takes as input a seller's cost-to-utility ratio $\gamma := \frac{c}{u}$, and outputs the fraction of the item purchased from the seller.
The allocation function is given by
$$f_r(\gamma):=
\begin{cases} \ln(e-\frac{\gamma}{r}) & \gamma< r(e-1)\\
0 & \gamma\ge r(e-1)
\end{cases}.$$
There is an associated payment function (per utility) $Q_{f_r}(\gamma)$ given by Myerson's lemma. The mechanism is simply choosing the $r$ to be the largest value such that the total payment is within the budget.

\begin{algorithm}[ht]
\SetAlgoLined
\SetKwInOut{Input}{Input}
\SetKwInOut{Output}{Output}
\Input{$(c_i,u_i)$ for $i\in[n]$, $B$.}
\SetAlgorithmName{Mechanism}~~
Let $r$ be the largest value such that $\sum_{i=1}^n Q_{f_r}(\frac{c_i}{u_i})\cdot u_i\le B$\;
For each seller $i$, buy $f_r(\frac{c_i}{u_i})$ fraction of the $i$-th item and pay $Q_{f_r}(\frac{c_i}{u_i})\cdot u_i$ to the seller.
 \caption{\textsc{AGN}}
 \label{mech:AGN}
\end{algorithm}

The following fact by elementary calculus is the key to understand the behavior of this mechanism, which says that $Q_{f_r}(\gamma)-\gamma$,~\ie, the ``information rent'' that the buyer overpays the seller, is a linear function in $f_r(\gamma)$ with $f_r(\gamma)$-intercept $1-1/e$.
\begin{fact}[Implicit in~\cite{anari2018budget}]\label{fact:AGN_line}
$Q_{f_r}(\gamma)=ref_r(\gamma)-r(e-1)+\gamma$.
\end{fact}
\begin{proof}
By definition, $Q_{f_r}(\gamma)=f_r(\gamma)\cdot \gamma+\int_{\gamma}^{r(e-1)} f_r(u)du$ for $\gamma\le r(e-1)$, and $Q_{f_r}(\gamma)=0$ for $\gamma\ge r(e-1)$. Taking the integration, we have that
\begin{align*}
    \int_{\gamma}^{r(e-1)} f_r(u)du&=\int_{\gamma}^{r(e-1)} \ln(e-\frac{u}{r})du \\
    &=\int_{\ln(e-\frac{\gamma}{r})}^{0} \underbrace{\ln\left(e-\frac{r(e-e^t)}{r}\right)}_{=t}d(r(e-e^t)) \\
    &=(r(e-e^t)\cdot t)|_{\ln(e-\frac{\gamma}{r})}^{0} - \int_{\ln(e-\frac{\gamma}{r})}^{0} r(e-e^t) dt \\
    &=(r(e-e^t)\cdot t)|_{\ln(e-\frac{\gamma}{r})}^{0}-r(et-e^t)|_{\ln(e-\frac{\gamma}{r})}^0 \\
    &=(r(1-t)\cdot e^t)|_{\ln(e-\frac{\gamma}{r})}^{0} \\
    &=r-r(1-\ln(e-\frac{\gamma}{r}))\cdot(e-\frac{\gamma}{r}) \\
    &=r-(1-f_r(\gamma))\cdot(re-\gamma).
\end{align*}
Therefore, this fact follows from
\begin{align*}
    Q_{f_r}(\gamma)&=f_r(\gamma)\cdot \gamma+\int_{\gamma}^{r(e-1)} f_r(u)du \\
    &=f_r(\gamma)\cdot \gamma+r-(1-f_r(\gamma))\cdot(re-\gamma) \\
    &=ref_r(\gamma)-r(e-1)+\gamma.
\end{align*}
\end{proof}

Now we prove the following hardness result for the $\textsc{AGN}$ mechanism.

\begin{theorem}\label{thm:hard_instance_for_agn}
For any bounded range of budgets, for any $\epsilon>0$, there is a hard market which satisfies the small-bidder assumption such that the \textsc{AGN} mechanism has competitive ratio at most $1-1/e+\epsilon$ for every budget in the range.
\end{theorem}
\begin{proof}

We first construct a hard market for arbitrary number of budgets $0<B_1<B_2<\dots<B_m$.
\paragraph{Construction} The market has $m$ buckets of sellers with unit utilities. For each $i\in[m]$, the $i$-th bucket has $\lambda_in$ sellers of the same cost $c_i$, and hence each seller $s$ in the $i$-th bucket has cost-per-utility $\gamma_s=c_i$. The $c_i$'s and $\lambda_i$'s are specified as follows.
Let $c_1:=B_1/n$ and $\lambda_1:=1$, and hence $\lambda_1c_1n=B_1$, and let $r_1$ be such that $\ln(e-\frac{c_1}{r_1})=1-\frac{1}{e}$. Let $c_2:=r_1(e-1)$, and $\lambda_2$ be such that $\frac{\lambda_1c_1+\lambda_2c_2}{\lambda_1c_1}=\frac{B_2}{B_1}$. Inductively, let $r_{i-1}$ be such that $\sum_{j=1}^{i-1}\lambda_j\cdot\ln(e-\frac{c_j}{r_{i-1}})=(1-\frac{1}{e})\sum_{j=1}^{i-1}\lambda_j$ and $c_i:=r_{i-1}(e-1)$, and $\lambda_i$ is such that $\frac{\sum_{j=1}^i\lambda_j c_j}{\sum_{j=1}^{i-1}\lambda_j c_j}=\frac{B_i}{B_{i-1}}$.

\paragraph{Analysis}
In the above market, we show by induction that $c_i$'s and $r_i$'s are monotone increasing.
For the base case, it follows by definition of $r_1$ that $r_1>c_1/(e-1)$, and hence, we have $c_2>c_1$ by definition of $c_2$. For the induction step, observe that when $i\ge 2$, $r_{i}>r_{i-1}$ since 
\begin{align*}
    \sum_{j=1}^{i}\lambda_j\cdot\ln(e-\frac{c_j}{r_{i}})&=(1-\frac{1}{e})\sum_{j=1}^{i}\lambda_j && \text{(By definition of $r_i$)}\\
    &>(1-\frac{1}{e})\sum_{j=1}^{i-1}\lambda_j \\
    &=\sum_{j=1}^{i-1}\lambda_j\cdot\ln(e-\frac{c_j}{r_{i-1}}) &&\text{(By definition of $r_{i-1}$)}\\
    &=\sum_{j=1}^{i}\lambda_j\cdot\ln(e-\frac{c_j}{r_{i-1}}) &&\text{(By definition of $c_i$)},
\end{align*}
and $r_{i}>r_{i-1}$ implies $c_{i+1}>c_{i}$ by definition of $c_i$.

By monotonicity of $c_j$, for budget $B_i$, the allocation rule $f_{r_i}$ does not select anything from $(\ge i+1)$-th bucket, because $c_{j}\ge r_i(e-1)$ for all $j\ge i+1$. Hence, $Q_{f_{r_i}}(\gamma_{s})=0$ for all $s$ in $(\ge i+1)$-th bucket. Next, we show that $\sum_{s\,\in\,\textrm{first $i$ buckets}} Q_{f_{r_i}}(\gamma_s)=B_i$ (which implies $\sum_{s=1}^n Q_{f_{r_i}}(\gamma_s)=B_i$). Specifically, we derive that
\begin{align*}
    \sum_{s\,\in\,\textrm{first $i$ buckets}} Q_{f_{r_i}}(\gamma_s)&=\sum_{s\,\in\,\textrm{first $i$ buckets}}(\gamma_s+r_ief_{r_i}(\gamma_s)-r_i(e-1)) \quad\quad\quad\text{(By Fact~\ref{fact:AGN_line})}\\
    &=\sum_{j=1}^i \lambda_j n\cdot (c_j+r_ief_{r_i}(c_j)-r_i(e-1))\\ &\quad\quad\quad\quad\quad\,\text{($j$-th bucket has $\lambda_{j}n$ sellers with cost-per-utility $c_j$)}\\
    &=B_i+\sum_{j=1}^i \lambda_j n\cdot (r_ief_{r_i}(c_j)-r_i(e-1)) \quad\text{(By definition of $\lambda_j$'s)}.
\end{align*}
Thus, it suffices to prove $\sum_{j=1}^i \lambda_j n\cdot (r_ief_{r_i}(c_j)-r_i(e-1))=0$. By rearranging, this is equivalent to
$\sum_{j=1}^i \lambda_j\cdot f_{r_i}(c_j)=(1-\frac{1}{e})\sum_{j=1}^i \lambda_j$, which in turn holds by definition of $r_i$.

Finally, observe that by definition of $f_r$, if a non-zero fraction of seller $s$'s item is allocated by $f_r$, and $\gamma_s>0$ (which holds for every seller in our market), then increasing $r$ of $f_{r}$ will strictly increase the allocated fraction of seller $s$'s item and hence the payment to seller $s$. Thus, it follows from $\sum_{s\in[n]} Q_{f_{r_i}}(\gamma_s)=B_i$ that given budget $B_i$, the allocation function selected by \textsc{AGN} is $f_{r_i}$. Therefore, the total utility achieved by \textsc{AGN} for budget $B_i$ is
\begin{align*}
    \sum_{s\,\in\,\textrm{first $i$ buckets}}f_{r_i}(\gamma_{s})&=\sum_{j=1}^i \lambda_jn\cdot f_{r_i}(c_j)\quad \text{($j$-th bucket has $\lambda_{j}n$ sellers with cost-per-utility $c_j$)}\\
    &=(1-\frac{1}{e})\sum_{j=1}^i \lambda_jn \quad\text{(By definition of $r_i$)}.
\end{align*}
Notice that $B_i$ is just enough to buy the first $i$ buckets of sellers for the non-IC optimum. Hence, the mechanism has competitive ratio $1-1/e$ for budget $B_i$ (the choice of $i\in[m]$ is arbitrary).

\paragraph{Extending to continuous range of budgets}
We remark that the $1-1/e$ barrier still holds even if the range of budgets is continuous. Because by our construction, $\frac{OPT_{k+1}}{OPT_k}=\frac{B_{k+1}}{B_k}$, where $OPT_k$ denotes the non-IC optimum of the market for budget $B_k$, we can choose $\frac{B_{k+1}}{B_k}=1+\epsilon$ such that the optimal value does not change much from $B_k$ to $B_{k+1}$. It follows that for any budget between $B_k$ and $B_{k+1}$, the mechanism can only have competitive ratio close to $1-1/e$.

\end{proof}

\section{Computing optimal budget-smoothed competitive ratios}\label{section:numerical_computing_ratios}
In this section, we formulate a (non-convex) mathematical program that computes the optimal budget-smoothed competitive ratio for any fixed distribution of $m$ budget perturbation factors (with respect to the largest budget $B$) $0<\rho_1<\dots<\rho_m=1$ (let $p_i$ denote the probability of the perturbation $\rho_i$) and numerically solve it for various distributions with small $m$. By Theorem~\ref{thm:worst_distribution}, it suffices to consider the worst-case Bayesian market, in which the $F(c)$-to-$cF(c)$ curve is a piecewise-linear function with non-decreasing slope and has at most $m+1$ pieces including the zero part. Let $F_1$ be the value of $F(c)$ where the curve switches from the zero part to the first linear function, and similarly, for $2\le i\le m$, let $F_i$ be the value of $F(c)$ where the curve switches from $(i-1)$-th linear function to $i$-th linear function. Moreover, for all $i\in[m]$, let $a_i$ be the slope of the $i$-th linear function of the curve. Clearly, the curve is fully determined by $F_i$'s and $a_i$'s. Henceforth, $a_i's$ and $F_i$'s are the only variables in our mathematical program, and all other quantities are functions of them, but to simplify notation, we will not explicitly write the arguments of the functions.
Now, let $F_{m+1}=1$ and let $c_i$ be such that $F(c_i)=F_i$ for $i\in[m+1]$. Let $y_i$ denote the value of the curve at $F_i$, then we have that for all $i\in[m+1]$
\[
    y_i=\sum_{j=1}^{i-1} a_j(F_{j+1}-F_{j}).
\]
To ensure non-decreasing
slope, we introduce the following constraints
\begin{equation}\label{eq:constraint}
\begin{split}
    &0\le F_1\le\dots\le F_m\le F_{m+1}=1,\\
    &0<a_1\le\dots\le a_m.
\end{split}
\end{equation}
Furthermore, without loss of generality, we assume that with the largest budget $B$, the optimal solution can buy all the items in the Bayesian market (in expectation), since otherwise, we can simply scale the CDF. Next, we calculate the total cost of the first $F(c)$ fraction of items for arbitrary $F_{i}\le F(c)\le F_{i+1}$. First, notice that for any $c_j\le c\le c_{j+1}$, $cF(c)-y_j=a_j(F(c)-F_j)$, and hence by letting $b_j:=a_jF_j-y_j$, it follows that $F(c)=\frac{b_j}{a_j-c}$. Integrating over $[c,c_j]$, we have that we have that $$\int_{c_j}^{c} F(x)dx=b_j\ln{\frac{a_j-c_j}{a_j-c}}.$$ 
Therefore, the total cost of the first $F(c)$-fraction of items for arbitrary $F_{i}\le F(c)\le F_{i+1}$ is
\begin{align*}
    \int_0^{F(c)} x dF(x)&=xF(x)|_0^c-\int_0^c F(x)dx \qquad\qquad\text{(Integration by parts)}\\
    &=cF(c)-\left(\left(\sum_{j=2}^{i}\int_{c_{j-1}}^{c_j} F(x)dx\right) + \int_{c_i}^{c} F(x)dx\right) \\
    &=(a_iF(c)-b_i)-\left(\left(\sum_{j=2}^{i}b_{j-1}\ln{\frac{a_{j-1}-c_{j-1}}{a_{j-1}-c_j}}\right)+b_i\ln{\frac{a_i-c_i}{a_i-c}}\right) \\
    &=a_iF(c)-b_i-\left(\left(\sum_{j=2}^{i}b_{j-1}\ln{\frac{a_{j-1}-c_{j-1}}{a_{j-1}-c_j}}\right)+b_i\ln{\frac{a_i-c_i}{a_i-\frac{a_iF(c)-b_i}{F(c)}}}\right) \\
    &=a_iF(c)-b_i-\left(\left(\sum_{j=2}^{i}b_{j-1}\ln{\frac{a_{j-1}-c_{j-1}}{a_{j-1}-c_j}}\right)+b_i\ln{\frac{a_i-c_i}{b_i}}+b_i\ln{F(c)}\right).
\end{align*}
In particular, $B=\int_0^1 x dF(x)=a_m-b_m-\sum_{j=1}^{m}b_j\ln{\frac{a_{j}-c_{j}}{a_{j}-c_{j+1}}}$, and now we can calculate the non-IC optimal utility $g_k$ achieved by budget $\rho_k B$. If $\int_0^{F_i} x dF(x)\le \rho_k B\le \int_0^{F_{i+1}} x dF(x)$, then we need to solve for $F(c)$ in the following equation
\[
    a_iF(c)-b_i-\left(\sum_{j=2}^{i}b_{j-1}\ln{\frac{a_{j-1}-c_{j-1}}{a_{j-1}-c_j}}+b_i\ln{\frac{a_i-c_i}{b_i}}+b_i\ln{F(c)}\right)=\rho_k B.
\]
Let $h_k:=\rho_k B+b_i+\sum_{j=2}^{i}b_{j-1}\ln{\frac{a_{j-1}-c_{j-1}}{a_{j-1}-c_j}}+b_i\ln{\frac{a_i-c_i}{b_i}}$. Then above equation is simplified to $a_iF(c)-b_i\ln{F(c)}=h_k$. This in turn is equivalent to solving $e^{-F(c)\cdot a_i/b_i}(-F(c)\cdot a_i/b_i)=-(a_i/b_i)e^{-h_k/b_i}$. To this end, we use the Lambert W function\footnote{Namely, the branches of the inverse relation of the function $f(w) = we^w$, see~\cite{corless1996lambertw} for example.}. Since $-F(c)\cdot a_i/b_i=-(cF(c)+b_i)/b_i=-1-cF(c)/b_i<-1$, we should use the branch $W_{-1}$, \ie, $-F(c)\cdot a_i/b_i=W_{-1}(-(a_i/b_i)e^{-h_k/b_i})$. Therefore, we can represent $g_k$ as following piecewise function
\[
    g_k=-\frac{b_i}{a_i}\cdot W_{-1}(-\frac{a_i}{b_i}\cdot e^{-h_k/b_i}),\,\textrm{if $\int_0^{F_i} x dF(x)\le \rho_k B\le \int_0^{F_{i+1}} x dF(x)$}.
\]
Next, we calculate the best achievable utility $f_k$ of a truthful-in-expectation mechanism with budget $\rho_k B$. As shown in Theorem~\ref{thm:worst_distribution}, the optimal truthful mechanism for the Bayesian market is a cutoff rule. By Myerson's payment rule, the budget spent by a cutoff rule at $c$ is $cF(c)$. Thus the best achievable utility for budget $\rho_k B$ is the $F(c)$ such that $cF(c)=\rho_k B$. If $c_iF_i\le \rho_k B\le c_{i+1}F_{i+1}$, then $cF(c)=a_iF(c)-b_i$, and hence, 
\[
f_k=(\rho_k B+b_i)/a_i,\,\textrm{if $c_iF_i\le \rho_k B\le c_{i+1}F_{i+1}$}.
\]
The final mathematical program is
\[
\min_{a_i,F_i}\sum_{k=1}^m p_k\cdot\frac{f_k}{g_k},\,\textrm{s.t. Eq.~\eqref{eq:constraint}}. 
\]
We solve this program numerically for various interesting distributions, and the results are summarized in Table~\ref{table:worst_case_ratio}. For the uniform distribution over two budgets $\rho B$ and $B$, we solve the program for $\rho\in[0.01,0.99]$ and present the optimal budget-smoothed competitive ratios in Figure~\ref{fig:two_budget}.

\begin{table}[h!]
\begin{minipage}{\columnwidth}
\begin{center}
\caption{\label{table:worst_case_ratio} Optimal budget-smoothed competitive ratios on various budget distributions}
\begin{tabular}{ c c } 
 \hline
{\bf  Budget perturbation distribution} & \textbf{Optimal budget-smoothed competitive ratio} \\
 \hline
 Uniform over $[1,10]$ & 0.64 \\
 Log-scale-uniform over $[1,8]$ & 0.65 \\
 Log-scale-uniform over $[1,512]$ & 0.67 \\
 Microworkers & 0.64 \\
 \hline
\end{tabular}
\end{center}
\footnotesize
We calculated tight budget-smoothed competitive ratios for several exemplary budget distributions. The last row of the table contains a real-world distribution,
which is is uniform over top 10 budgets spent on Microworkers~\cite{microworkersdata}.
The data are provided in Table~\ref{table:data}.
\end{minipage}
\end{table}

\begin{table}[h!]
\begin{minipage}{\columnwidth}
\begin{center}
\caption{\label{table:data} Top 10 budgets spent on Microworkers}
\begin{tabular}{ c c } 
 \hline
{\bf Dataset } & \textbf{Budget perturbations} \\
 \hline
 Microworkers & $\{0.124,0.126,0.154,0.172,0.236,0.281,0.299,0.544,0.625,1\}$ \\
 \hline
\end{tabular}
\end{center}
\footnotesize This table contains the real-world budget data of Microworkers, which are normalized and rounded for privacy.
\end{minipage}
\end{table}

\begin{figure}
\centering
    \includegraphics[scale=0.5]{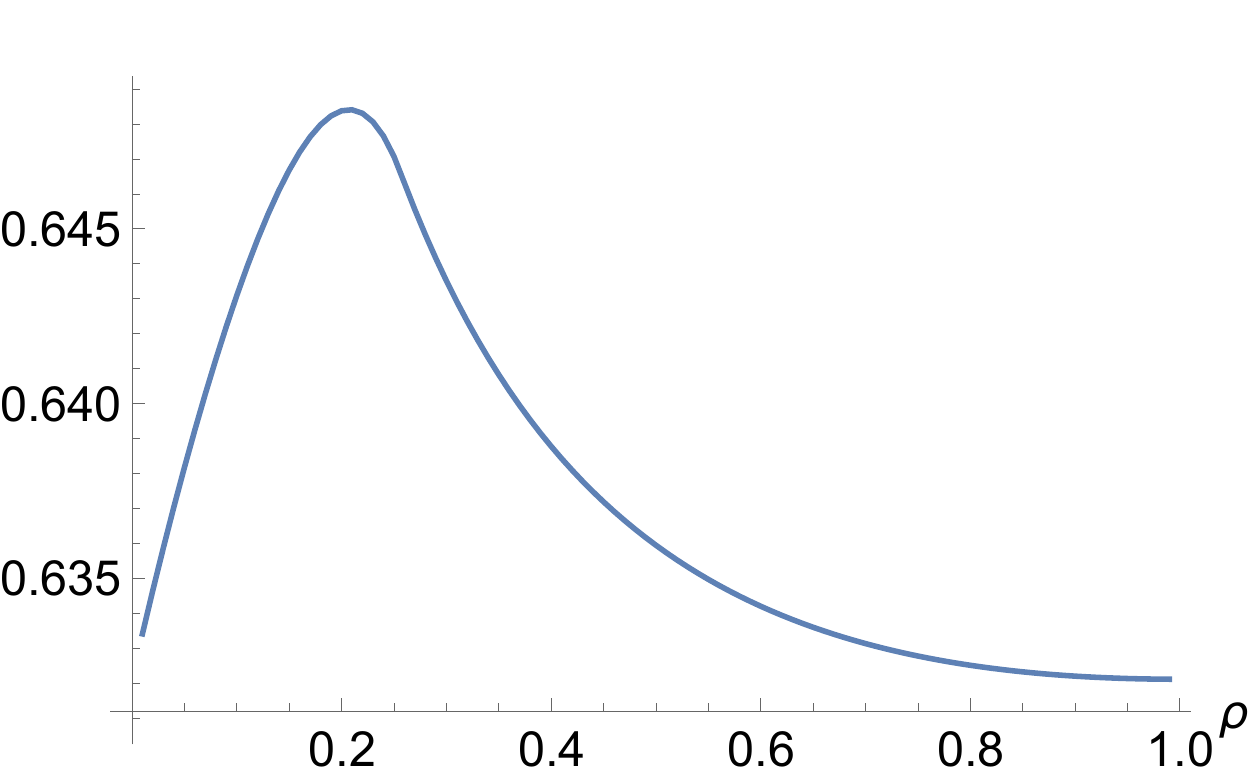}
    \caption{Optimal average competitive ratios for budgets $\rho B$ and $B$.}
    \label{fig:two_budget}
\end{figure}

We remark that the numerical results are interesting on both positive and negative sides --- On the positive side, the optimal budget-smoothed competitive ratios have modest but non-negligible improvements over $1-1/e$ for many interesting distributions. On the negative side, we pin down the worst-case markets for these distributions which remain significantly hard, which might provide new insights on how to model the structure of beyond-worst-case markets to exclude these markets.

In light of Theorem~\ref{thm:lower_bound}, we know there is a limit for the positive side of budget-smoothed analysis. That said, it is noteworthy that because the mathematical program for solving the optimal budget-smoothed competitive ratio is non-convex, we are only able to compute the optimal ratios after discretizing the budget distributions with a limited number of budgets, and therefore, the positive results may still have a lot of room to improve. This leaves an interesting open problem: close the gap by identifying the best budget distribution and solving the optimal budget-smoothed competitive ratio for the best budget distribution.

\end{document}